\documentclass[]{article}
\usepackage{etex}
\reserveinserts{28}

\usepackage{amscd, amsmath, amssymb, epic, eepic,  delarray,stmaryrd,mathbbol,amsbsy,amsfonts}

\usepackage[latin1]{inputenc}
\usepackage{graphicx,epsfig,epsf,epsfig}
\usepackage{pstricks,pst-node,rotating}
\usepackage{color}
\usepackage{tikz}
\usepackage{pgf}
\usetikzlibrary{arrows,automata}
\tikzstyle{my loopup}=[->, to path={
 .. controls +(50:1) and +(100:1) .. (\tikztotarget) \tikztonodes}]

\tikzstyle{my loopdown}=[->, to path={
 .. controls +(-110:1) and +(-30:1) .. (\tikztotarget) \tikztonodes}]

\tikzstyle{my arcdown}=[->, to path={
 .. controls +(-100:1) and +(-120:1) .. (\tikztotarget) \tikztonodes}]

\tikzstyle{my arcup}=[->, to path={
 .. controls +(100:1) and +(120:1) .. (\tikztotarget) \tikztonodes}]

\usetikzlibrary{matrix}

\newcommand{\Isa}[1]{\textcolor{blue}{#1}}
\newcommand{\Marc}[1]{\textcolor{red}{#1}}

\newcommand\Omit[1]{}

\begin{document}

%%%%%%%%%%%%%%%%
% Environnement %
%%%%%%%%%%%%%%%%% 

%setcounter{secnumdepth}{7}

\newtheorem{definition}{Definition}[section]
\newtheorem{notation}[definition]{Notation}
\newtheorem{lemma}[definition]{Lemma}
\newtheorem{proposition}[definition]{Proposition}
\newtheorem{theorem}[definition]{Theorem}
\newtheorem{corollary}[definition]{Corollary}
\newtheorem{fact}[definition]{Fact}
\newtheorem{example}[definition]{Example}
\newtheorem{constraint}[definition]{Constraint}
\newtheorem{convention}[definition]{Convention}
\newtheorem{remark}[definition]{Remark}

\newenvironment{proof}{\noindent {\bf Proof}}{\hfill$\Box$
\medskip
}

\newtheorem{axiom}{Axiom}

\newcommand{\Introduction}{\section*{Introduction}
                           \addcontentsline{toc}{chapter}{Introduction}}

\def\sl#1{\underline{#1}}
\def\cl#1{{\mathcal{#1}}}

\newcommand{\df}[1]{#1\!\!\searrow}
\newcommand{\udf}[1]{#1\!\!\nearrow}

\newcommand{\rmq}{\hspace{-5mm}{\bf Remark:} ~}
\newcommand{\lpar}{\par\noindent}

%###############################################################
%##### ENVT POUR CONSTRUIRE DES ARBRES DE PREUVE ###############
%###############################################################

\newenvironment{infrule}{\begin{array}{c}}{\end{array}}

\def\rulename#1{({\sc #1})}

\def\et{\hskip 1.5em \relax}

\def\nm{\vspace{-1mm} \\  \hspace{-0.6cm}}
\def\nom{\vspace{-1mm} \\  \hspace{-1.1cm}}
\def\nomq{\vspace{-1mm} \\  \hspace{-1.3cm}}
\def\nomter{\vspace{-1mm} \\  \hspace{-1.4cm}}
\def\sansnom{\vspace{-2mm} \\ \hspace{-0.2cm}}
\def\nombis{\vspace{-1mm} \\  \hspace{-1.8cm}}
\def\decal{\vspace{-2mm} \\ \hspace{0.4cm}}

\def\imp{ ~ \hrulefill \\}

\def\ligne{\\[5mm]}

\def\Ligne{\\[7mm]}

%###########################################

\newenvironment{ex.}{
        \medskip
        \noindent {\bf Example}}
        {\hfill$\Box$ \medskip}
\newenvironment{Enonce}[1]{
        \begin{description}
        \item[{\bf #1~:}]
        \mbox{}
}{
        \end{description}
}

%%%%%%%%%%%%%%%%%%
% Macro generale %
%%%%%%%%%%%%%%%%%%
%\def\Doubleunion#1#2#3{\displaystyle{\bigcup_{\scriptstyle #1 \atop
%\scriptstyle #2}}#3}
\def\Doubleunion#1#2#3{\displaystyle{\bigcup_{#1}^{#2} {#3}}}
\def\Union#1#2{\displaystyle{\bigcup_{#1} #2}}
\def\Intersect#1#2{\displaystyle{\bigcap_{#1} #2}}
\def\Coprod#1#2{\displaystyle{\coprod_{#1} #2}}
\def\Produit#1#2#3{\displaystyle{\prod_{#1}^{#2} {#3}}}
\def\Conj#1#2{\displaystyle{\bigwedge_{#1} \!\!\!\!\!#2}}
\def\Disj#1#2{\displaystyle{\bigvee_{#1} \!\!\!\!\!#2}}

\newcommand{\Nat}{I\!\!N}
\newcommand{\Int}{Z\!\!\! Z}

%%r\`egle d'inf\'erence
\def\infer#1#2{\mbox{\large  ${#1} \frac {#2}$ \normalsize}}

\def\sem#1{[\![ #1 ] \!]}

%%%%%%%%%%%%%%%%%%%%%%%%%%%%%%
% Mot vide sur tout alphabet %
%%%%%%%%%%%%%%%%%%%%%%%%%%%%%%

\newcommand{\mv}{\varepsilon}

%%%%%%%%%%%%%%%%%%%%
% Relation binaire %
%%%%%%%%%%%%%%%%%%%%
\def\M#1#2{M_{\scriptstyle #1 \atop \scriptstyle #2}}
\def\H#1#2{H_{\scriptstyle #1 \atop \scriptstyle #2}}
\newcommand{\Ro}{{\cal{R}}}
\newcommand{\Mo}{{\cal{M}}}
\newcommand{\Ho}{{\cal{H}}}
\newcommand{\Eo}{{\cal{E}}}
\newcommand{\ssucc}{\succ \!\! \succ}

%%%%%%%%%%%%
% Identite %
%%%%%%%%%%%%

\def\id#1{\underline{#1}}

%============================= Title and Abstract

\title{Dual Logic Concepts based on Mathematical Morphology in Stratified Institutions: Applications to Spatial Reasoning}

\author{Marc Aiguier$^1$ and Isabelle Bloch$^2$ \\
1. MICS, CentraleSupelec, Universit\'e Paris Saclay, France \\ {\it marc.aiguier@centralesupelec.fr} \\
2. LTCI, T\'el\'ecom ParisTech, Universit\'e Paris Saclay, Paris, France \\ {\it isabelle.bloch@telecom-paristech.fr}}
\date{}

\maketitle

\begin{abstract}
Several logical operators are defined as dual pairs, in different types of logics. Such dual pairs of operators also occur in other algebraic theories, such as mathematical morphology. Based on this observation, this paper proposes to define, at the abstract level of institutions, a pair of abstract dual and logical operators as morphological erosion and dilation. Standard quantifiers and modalities are then derived from these two abstract logical operators. These operators are studied both on sets of states and sets of models. To cope with the lack of explicit set of states in institutions, the proposed abstract logical dual operators are defined in an extension of institutions, the stratified institutions, which take into account the notion of open sentences, the satisfaction of which is parametrized by sets of states. 
%Dealing explicitly with sets of states, an internal fix-point logic is developed in the stratified institutions framework by introducing two abstract fix-point operators defined again by morphological erosion and dilation. 
A hint on the potential interest of the proposed framework for spatial reasoning is also provided.
\end{abstract}

\noindent
{\small {\bf Keywords:} Stratified institutions, mathematical morphology, dual operators, dilation, erosion, states, spatial reasoning.}

%================================ Start Text

%%

\section{Introduction}

There exists a profusion of logics but all of them satisfy the same structure defined by a syntax, a semantics and a calculus. Syntax gives both the language (signatures) and the formal rules that define well-formed formulas and theories. Semantics, so-called model theory,  gives the mathematical meaning of all these syntactic notions, among others the rules that  associate truth values to formulas. Finally, calculus, so-called proof theory, gives the inference rules that govern the reasoning and thus translate semantics into syntax as correctly as possible. To cope with the explosion of logics, a categorical abstract model-theory, the theory of institutions~\cite{Dia08,GB92}, has been proposed, that generalizes Barwise's ``Translation Axiom''~\cite{Bar74}. Institutions then define both syntax and semantics of logics at an abstract level,
%with the ambition to be as much as possible at the level of abstraction 
independently of commitment to any particular logic. Later, institutions have been extended to propose a syntactic approach to truth~\cite{Dia06,Dia08,FS88,Mes89}. For the sake of generalization, in institutions signatures are simply defined as objects of a category and formulas built over signatures are simply required to form a set. All other contingencies such as inductive definition of formulas are not considered. However, the reasoning (both syntactic and semantic) is defined by induction on the structure of formulas. Indeed, usually, formulas are built from ``atomic'' formulas by applying iteratively operators such as connectives, quantifiers or modalities. What we can then observe is that most of these logical operators come through dual pairs (conjunction and disjunction $\wedge$ and $\vee$, quantifiers $\forall$ and $\exists$, modalities $\Box$ and $\Diamond$). 

When looking at the algebraic properties of mathematical morphology~\cite{BHR07,Ser82} on the one hand, and of all these dual operators on the other hand, several similarities can be shown, and suggest that links between institutions and mathematical morphology are worth to be investigated. This has already been done in the restricted framework of modal propositional logic~\cite{Bloch02}. In~\cite{Bloch02}, it was then shown that modalities $\Box$ and $\Diamond$ can be defined as morphological erosion and dilation. The interest is, based on properties of morphological operators, that this leads to a set of axioms and inference rules which are de facto sound. In this paper, we propose to extend this work by defining, at the abstract level of institutions, a pair of abstract operators as morphological erosion and dilation. We will then show how to obtain standard quantifiers and modalities from these two abstract operators. 

In mathematical morphology, erosion and dilation are operations that work on lattices, for instance on sets. Thus, they can be applied to formulas by identifying formulas with sets. We have two ways of doing this, either given a model $M$ identifying a formula $\varphi$ by the set of states $\eta$ that satisfy $\varphi$ and classically noted $M \models_\eta \varphi$, or identifying $\varphi$ by the set of models that satisfy it. As usual in logic, our abstract dual operators based on morphological erosion and dilation will be studied both on sets of states and sets of models. The problem is that institutions do not explicit, given a model $M$, its set of states. This is why we will define our abstract logical dual operators based on erosion and dilation in an extension of institutions, the stratified institutions~\cite{AD07}. Stratified institutions have been defined in~\cite{AD07} as an extension of institutions to take into account the notion of open sentences, the satisfaction of which is parametrized by sets of states. For instance, in first-order logic, the satisfaction is parametrized by the valuation of unbound variables, while in modal logics it is further parametrized by possible worlds. Hence, stratified institutions allow for a uniform treatment of such parametrizations of the satisfaction relation within the abstract setting of logics as institutions.  

%Dealing explicitly with sets of states, we allow further developing an internal fix-point logic in the stratified institutions framework by introducing two abstract fix-point operators $\mu x$ and $\nu x$ that denote respectively the least and the greatest fix-point operators. Here again, we will show that both fix-point operators can be defined by morphological erosion and dilation. 

Another interest of the approach proposed in this paper is that mathematical morphology provides tools for spatial reasoning. Until now, mathematical morphology has been used mainly for quantitative representations of spatial relations, or semi-qualitative ones, in a fuzzy set framework (see e.g.~\cite{IB:IVC-05}). For qualitative spatial reasoning, several symbolic and logical approaches have been developed~(see e.g. \cite{Handbook07,AB02,Lig11}), but mathematical morphology has not been much used in this context to our knowledge. In this paper, inspired by the work that was done in~\cite{Bloch02,IB:IJAR-06,BHR07,BL02} in the propositional and modal logic framework, we show how logical connectives based on morphological operators can be used for symbolic representations of spatial relations. Indeed, spatial relations are a main component of spatial reasoning~\cite{Handbook07}, and several frameworks have been proposed to model spatial relations and reason about them in logical frameworks (see e.g.~\cite{BD07,CF97,CBGG97,RCC92,BB07} for topological relations, \cite{Lig11,MM15} for directional relations, \cite{RN07} for constraint based techniques for topology, distances and directions, and~\cite{IB:IVC-05,IB:IJAR-06} for semi-qualitative representations in the framework of fuzzy sets).  Since it is usual to introduce uncertainty in qualitative spatial reasoning, we propose to extend our abstract logical connectives based on erosion and dilation to the fuzzy case. This first requires to develop fuzzy reasoning in stratified institutions. Fuzzy (or many-valued) reasoning has an institutional semantics~\cite{Dia13,Dia14}. The approach proposed here is substantially similar to that proposed in~\cite{Dia13}, although developed in stratified institutions.  

\medskip
The paper is organized as follows. Section~\ref{stratified institutions} reviews some concepts, notations and terminology about institutions and stratified institutions which are used in this work. In Section~\ref{extensions} we propose to define abstractly the important concept of Boolean connectives, quantifiers, and fuzzy reasoning in stratified institutions. 
%\Marc{In Section~\ref{abstract modalities}, we continue this process of abstraction by giving semantics for modalities independently \Isa{of} any stratified institution, nevertheless satisfying some constraints on the structure of models.} 
Section~\ref{dual connectives} introduces a new way to build dual operators from the notion of morphological erosion and dilation operators. We study two ways to build such dual operators. We first define them from morphological dilation and erosion of formulas based on a structuring element, and then as algebraic erosion and dilation over the lattice of formulas. This last point allows us to define modalities when they are interpreted topologically as algebraic erosion and dilation. Finally, in Section~\ref{qualitative spatial reasoning}, we show how these modalities can be interpreted for abstract spatial reasoning using qualitative representations of spatial relationships derived from mathematical morphology.    

\section{Stratified institutions}
\label{stratified institutions}

The notions introduced here make use of basic notions of category theory (category, functors, natural transformations, etc.). We do not present these notions in these preliminaries, but interested readers may refer to textbooks such as~\cite{BW90,McL71}. 

\subsection{Institutions}

Let us start by recalling the definition of institutions, over which stratified institutions are defined as an extension, by introducing the notion of states for models.

\begin{definition}[Institution]\label{def-instit}
An {\bf institution} ${\cal I} = (Sig,Sen,Mod,\models)$ consists of
\begin{itemize}
\item a category $Sig$, objects of which are called {\em signatures} and are denoted $\Sigma$,
\item a functor $Sen : Sig \rightarrow Set$ giving for each signature $\Sigma$ a set $Sen(\Sigma)$, elements of 
which are called {\em sentences}, 
\item a contravariant functor $Mod : Sig^{op} \rightarrow Cat$ giving~\footnote{Standardly in category theory, $Sig^{op}$ is the opposite of $Sig$ by reversing morphisms.} for each 
signature a category, objects and arrows of which are called {\em
$\Sigma$-models} and {\em $\Sigma$-morphisms} respectively, and 
\item a $Sig$-indexed family of relations $\models_\Sigma \subseteq Mod(\Sigma) \times Sen(\Sigma)$ called 
{\em satisfaction relation}, such that the following property, called the {\em satisfaction condition}, holds:\\
$\forall \sigma : \Sigma \rightarrow \Sigma',~\forall M' \in
Mod(\Sigma'),~\forall \varphi \in Sen(\Sigma)$,
$$M' \models_{\Sigma'} Sen(\sigma)(\varphi) \Leftrightarrow
Mod(\sigma)(M') \models_\Sigma \varphi$$
\end{itemize}
\end{definition}

\begin{notation}
The functor $Mod$ can be extended to formulas. Hence, given a signature $\Sigma$ and two formulas $\varphi,\psi \in Sen(\Sigma)$, we note:
\begin{itemize}
\item $Mod(\varphi) = \{M \in Mod(\Sigma) \mid M \models_\Sigma \varphi\}$, 
\item $\varphi \models \psi \Longleftrightarrow Mod(\varphi) \subseteq Mod(\psi)$, and
\item $\varphi \equiv \psi \Longleftrightarrow Mod(\varphi) = Mod(\psi)$.
\end{itemize}
\end{notation}

\begin{example}
\label{examples of institutions}
The following examples of institutions are of particular importance both in computer science and in this paper. Many other examples can be found in the literature (\emph{e.g.}~\cite{Dia08,GB92,Tar99}). 

\begin{description}
\item[Propositional Logic (PL)] The category of signatures is $Set$, the category of sets and
functions. \\
Given a signature $P$, the set of $P$-sentences is the least set
of sentences finitely built over propositional variables in $P$ and Boolean connectives in $\{\neg,\vee,\wedge,\Rightarrow\}$. Given a signature morphism $\sigma:P\to P'$, $Sen(\sigma)$ translates
$P$-formulas to $P'$-formulas by renaming propositional
variables according to $\sigma$.\\
Given a signature $P$, the category of $P$-models is 
%the category of mappings $\nu:P\to\{0,1\}$ 
$(\{0,1\}^P, \leq)$ such that $0$ and $1$ are the usual truth values, and $\leq$ is a partial ordering such that $\nu \leq \nu'$ iff $\forall p \in P, \nu(p) \leq \nu'(p)$.
% with identities as morphisms. 
Given a signature morphism $\sigma : P \rightarrow P'$, the forgetful functor $Mod(\sigma)$ maps a $P'$-model $\nu'$ to the $P$-model $\nu=\nu'\circ\sigma$. \\
Finally, satisfaction is the usual propositional satisfaction.

\item[Many-sorted First Order Logic (FOL)] Signatures are triplets $(S,F,P)$ where $S$ is a set
of sorts, and $F$ and $P$ are sets of function and predicate names respectively, both with arities in $S^\ast\times S$ and $S^+$ respectively.\footnote{$S^+$ is the set of all non-empty sequences of elements in $S$ and $S^\ast=S^+\cup\{\epsilon\}$ where $\epsilon$ denotes the empty sequence.} Signature morphisms $\sigma: (S,F,P) \rightarrow (S',F',P')$ consist of
three functions  between sets of sorts, sets of functions and sets of predicates respectively, the last two preserving arities.\\
Given a signature $\Sigma=(S,F,P)$, the $\Sigma$-atoms are $p(t_1,\ldots,t_n)$ where $p:s_1 \times \ldots \times s_n \in P$ and $t_i \in T_F(X)_{s_i}$ ($1 \leq i \leq n$, $s_i\in S$)~\footnote{$T_F(X)_s$ is the term algebra of sort $s$ built over $F$ with sorted variables in a given set $X$.}. The set of $\Sigma$-sentences is the least set of formulas built over the set of $\Sigma$-atoms by finitely applying Boolean connectives in $\{\neg,\vee,\wedge,\Rightarrow\}$ and the quantifiers $\forall$ and $\exists$.
% \Marc{J'ai retir\'e les \'equations dans la d\'efinition des formules atomiques qui demandaient de faire un quotient dans la d\'efinition du mod\`ele pour la proposition 4.8. J'ai aussi ajout\'e l'ensemble des connecteurs propositionnels ainsi que le quantificateur $\exists$.}. 
Given a signature morphism $\sigma : \Sigma \to \Sigma'$, $Sen(\sigma)$ is the mapping defined by renaming functions and predicates according to $\sigma$. \\
Given a signature $\Sigma = (S,F,P)$, a $\Sigma$-model
${\cal M}$ is a family ${\cal M} = (M_s)_{s \in S}$ of sets (one for every $s \in S$), each one equipped with a function $f^{\cal M} : M_{s_1} \times
\ldots \times M_{s_n} \rightarrow M_s$ for every $f:s_1 \times \ldots
\times s_n \rightarrow s  \in F$ and with a n-ary relation $p^{\cal M}
\subseteq M_{s_1} \times \ldots \times M_{s_n}$ for every $p:s_1 \times
\ldots \times s_n \in P$. A model morphism $\mu : \mathcal{M} \to \mathcal{M}'$ is a mapping $\mu : M \to M'$ that preserves sorts (i.e. $\mu(M_s) \subseteq M'_s$ for each $s \in S$) such that for every $f : s_1 \times \ldots \times s_n \to s \in F$ and every $(a_1,\ldots,a_n) \in M_{s_1} \times \ldots \times M_{s_n}$, $\mu(f^\mathcal{M}(a_1,\ldots,a_n)) = f^{\mathcal{M}'}(\mu(a_1),\ldots,\mu(a_n))$, and for every $p : s_1 \times \ldots \times s_n \in P$ and every $(a_1,\ldots,a_n) \in M_{s_1} \times \ldots \times M_{s_n}$, $(a_1,\ldots,a_n) \in p^\mathcal{M} \Longrightarrow (\mu(a_1),\ldots,\mu(a_n)) \in p ^{\mathcal{M}'}$. \\ Given a signature morphism $\sigma : \Sigma = (S,F,P) \rightarrow \Sigma' = (S',F',P')$ and a $\Sigma'$-model ${\cal M}'$,
$Mod(\sigma)({\cal M}')$ is the $\Sigma$-model ${\cal M}$ defined for
every $s \in S$ by $M_s = M'_s$, and for every function name $f \in F$
and every predicate name $p \in P$, by $f^{\cal M} = \sigma(f)^{{\cal 
M}'}$ and $p^{\cal M} = \sigma(p)^{{\cal M}'}$. \\
Finally, satisfaction is the usual first-order 
satisfaction. 

\item[Modal Propositional Logic (MPL)] The category of signatures is the same as {\bf PL}. 
For each set $P$, the $P$-sentences are formed from the elements of
$P$ by closing under Boolean connectives and unary modal connectives $\Box$ (necessity) and
$\Diamond$ (possibility).
A model $(I,W,R)$ for a signature $P$, called
\emph{Kripke model}, consists of 
\begin{itemize}
\item an index set $I$, 
\item a family $W = \{ W^i \}_{i\in I}$ of ``possible worlds", which are
functions from $P$ to $\{ 0,1 \}$ (or equivalently subsets of $P$), 
\item an ``accessibility" relation $R \subseteq I \times I$. 
\end{itemize}  
A model homomorphism $h : (I,W,R) \to (I',W',R')$ consists of
a function $h : I \to I'$ which preserves the accessibility
relation, i.e. $(i,j)\in R$ implies $(h(i),h(j)) \in R'$,
and such that $W^i \subseteq {W'}^{h(i)}$ for each $i\in I$. Given a signature morphism $\sigma: \ P \to P'$ and a $P'$-model $(I',W',R')$, $Mod(\sigma)((I',W',R'))$ is the $P$-model $(I,W,R)$ such that $I = I'$, $R = R'$ and $W^i = \{\nu' \circ \sigma \mid \nu' \in W'^{h(i)}\}$ for each $i \in I$. 
\\
The satisfaction of $P$-sentences by the Kripke $P$-models,
$(I,W,R) \models_P \varphi$, is defined by $(I,W,R)\models^i_P \varphi$ for each
$i\in I$, where $\models^i_P$ is defined by induction on the structure
of the sentences as follows: 
\begin{itemize}
\item $(I,W,R) \models^i_P p$ iff $p \in W^i$ for each $p \in P$, 
\item $(I,W,R) \models^i_P \varphi_1 \wedge \varphi_2$ iff
$(I,W,R) \models^i_P \varphi_1$ and $(I,W,R) \models^i_P \varphi_2$; and similarly
for the other Boolean connectives,
\item $(I,W,R) \models^i_P \Box \varphi$ iff $(I,W,R) \models^j_P \varphi$ for each $j$
such that $(i,j)\in R$, and 
\item $\Diamond \varphi$ is the same as $\neg \Box \neg \varphi$. 
\end{itemize}

\item[Topological MPL (TMPL)] In {\bf MPL}, the modalities $\Box$ and $\Diamond$ are interpreted relationally (i.e. in Kripke models). Here, they will be interpreted topologically. Hence, the category of signatures and the functor $Sen$ are the same as {\bf MPL}. Conversely, given a signature $P$, a $P$-model $M$ is a topological space $(X,\tau)$ equipped with a valuation function $\nu : P \to \mathcal{P}(X)$.~\footnote{We follow the definition given in~\cite{BB07}.} Such models are called {\em topos-models}. A model morphism $h : (X,\tau,\nu) \to (X',\tau',\nu')$ is a continuous mapping such that for every $p \in P$, $ h(\nu(p)) \subseteq \nu'(p)$. Given a signature morphism $\sigma : P \to P'$ and a $P'$-model $(X',\tau',\nu')$, $Mod(\sigma)((X',\tau',\nu'))$ is the $P$-model $(X,\tau,\nu)$ such that $X = X'$, $\tau = \tau'$ and $\nu = \nu' \circ \sigma$. \\ The satisfaction of sentences by the topological models, $(X,\tau,\nu) \models_P \varphi$, is defined by $(X,\tau,\nu)\models^x_P \varphi$ for each $x \in X$, where $\models^x_P$ is defined by induction on the structure of the sentences as follows: 
\begin{itemize}
\item $(X,\tau,\nu) \models^x_P p$ iff $x \in \nu(p)$ for each $p \in P$, 
\item $(X,\tau,\nu) \models^x_P \varphi_1 \wedge \varphi_2$ iff
$(X,\tau,\nu) \models^x_P \varphi_1$ and $(X,\tau,\nu) \models^x_P \varphi_2$, and similarly
for the other Boolean connectives,
\item $(X,\tau,\nu) \models^x_P \Box \varphi$ iff there exists $O \in \tau$ s.t. $x \in O$ and $(X,\tau,\nu) \models^y_P \varphi$ for each $y \in O$, and 
\item $\Diamond \varphi$ is the same as $\neg \Box \neg \varphi$.
\end{itemize}
Hence, $\Box$ and $\Diamond$ are interpreted as both topological notions of interior and closure, respectively.

\item[Metric MPL (MMPL)] Here, modalities will be interpreted in a metric space. The institution {\bf MMPL} has the same signatures and sentences as {\bf MPL} and {\bf TMPL}. Conversely, given a signature $P$, a $P$-model is a metric space $(X,d)$ equipped with a valuation function $\nu : P \to \mathcal{P}(X)$. Such models are called {\em metric models}. A model morphism $h : (X,d,\nu) \to (X',d',\nu')$ is a continuous mapping such that for every $p \in P$, $h(\nu(p)) \subseteq \nu'(p)$. Given a signature morphism $\sigma : P \to P'$ and a $P'$-model $(X',d',\nu')$, $Mod(\sigma)((X',d',\nu'))$ is the $P$-model $(X,d,\nu)$ such that $X = X'$, $d = d'$ and $\nu = \nu' \circ \sigma$. \\ The satisfaction of sentences by metric models $(X,d,\nu) \models_P \varphi$ is defined by $(X,d,\nu) \models^x_P \varphi$ for each $x \in X$, where $\models^x_P$ is defined by induction on the structure of the sentences as follows: 
\begin{itemize}
\item basic sentences and Boolean connectives are satisfied standardly;
\item $(X,d,\nu) \models^x_P \Box \varphi$ iff $\exists \varepsilon > 0$, $\forall y \in X$, $d(x,y) < \varepsilon \Rightarrow (X,d,\nu) \models^y_P \varphi$;
\item $\Diamond \varphi$ is the same as $\neg \Box \neg \varphi$.
\end{itemize}
\end{description}
\end{example}

\subsection{Stratified institutions}

Stratified institutions refine institutions by introducing the notion of states for models. Hence, each model $M$ is equipped with a set $\sem{M}$, elements of which are called states, such as possible worlds for Kripke models. \\ The definition of stratified institutions given in Definition~\ref{def:stratified} slightly improves the original one in~\cite{AD07} by considering a concrete category to equip models with states rather than the category of sets. This is motivated by the different applications developed in this paper such as the extensions of stratified institutions to modalities or to qualitative spatial reasoning, which require to consider in the first case sets equipped with binary relations, and in the second one topological or metric spaces.

\begin{definition}[Stratified institution]
\label{def:stratified}
A {\bf stratified institution} consists of:
\begin{itemize}
\item a category $Sig$ of signatures;
\item a sentence functor $Sen :  Sig \to Set$;
\item a model functor $Mod : Sig^{op} \to Cat$; 
\item a ``stratification'' $\sem{\_}$ which consists of a functor
$\sem{\_}_\Sigma : Mod(\Sigma) \to \mathcal{C}$ for
each signature $\Sigma \in Sig$ ({\bf states of models}) where $\mathcal{C}$ is a concrete category (i.e. $\mathcal{C}$ is equipped with a faithful functor $\mathcal{U} : \mathcal{C} \to Set$), and a natural transformation $\sem{\_}_\sigma : \sem{\_}_{\Sigma'} \to \sem{\_}_\Sigma \circ Mod(\sigma)$ for each signature morphism $\sigma : \Sigma \to \Sigma'$ such that $\mathcal{U}(\sem{M'}_\sigma)$ is surjective for each $M' \in Mod(\Sigma')$ (and then by standard results in the category theory, $\sem{M'}_\sigma$ is an epimorphism in $\mathcal{C}$)~\footnote{In many concrete categories of interest the converse is also true. However, this does not hold in general.}. To simplify the notations and when this does not raise ambiguities,
% in the paper, 
we use in the rest of this paper the notation $\sem{M}_\Sigma$, given a signature $\Sigma$ and a model $M \in Mod(\Sigma)$, to denote both the object in the concrete category $\mathcal{C}$ and the underlying set $\mathcal{U}(\sem{M}_\Sigma)$. Similarly, given a signature morphism $\sigma : \Sigma \to \Sigma'$ and a $\Sigma'$-model $M'$, we will use the notation $\sem{M'}_\sigma$ to denote both the morphism $\sem{M'}_\sigma$ in $\mathcal{C}$ and the mapping $\mathcal{U}(\sem{M'}_\sigma)$ in $Set$; 
%\Isa{j'ai mis des $U$ droits partout pour que ce soit homogene. C'est bon ?} \Marc{Oui, tout \`a fait.}
\item a satisfaction relation between models and sentences which is
parametrized by model states, $M \models^\eta_\Sigma \varphi$
where $\eta \in \sem{M}_\Sigma$ such that, $\forall \sigma : \Sigma \to \Sigma', \forall M \in Mod(\Sigma'), \forall \eta \in \sem{M}_{\Sigma'}, \forall \varphi \in Sen(\Sigma)$, the two following
properties are equivalent: 
\begin{enumerate}
\item $Mod(\sigma)(M) \models^{\sem{M}_\sigma (\eta)}_\Sigma \varphi$, 
\item $M \models^\eta_{\Sigma'} Sen(\sigma)(\varphi)$.
\end{enumerate} 
\end{itemize}
Then, we can define for every $\Sigma \in Sig$, the
satisfaction relation $\models_\Sigma \subseteq Mod(\Sigma)
\times Sen(\Sigma)$ as follows:
$$M \models_\Sigma \varphi \mbox{ \ if and only if \ } M
\models^\eta_\Sigma \varphi
\mbox{ \ for all \ } \eta \in \sem{M}_\Sigma .$$
\end{definition}

\begin{notation}
Given a signature $\Sigma \in Sig$, a model $M \in Mod(\Sigma)$ and a formula $\varphi \in Sen(\Sigma)$, we note $\sem{M}_\Sigma(\varphi) = \{\eta \in \sem{M}_\Sigma \mid M \models^\eta_\Sigma \varphi\}$.
\end{notation}

\begin{example}
\label{propositional logic}
{\bf PL} is the stratified institution with $Set$ as concrete category and $\sem{\nu}_P = \mathbb{1}$ ($\mathbb{1}$ is any singleton up to isomorphism) for each
set $P$ of propositional variables and each $P$-model $\nu$. 
\end{example}

\begin{example}[Internal stratification~\cite{AD07}]
\label{internal stratification}
In any institution $\mathcal{I} = (Sig,Sen,Mod,\models)$, we can define the stratified institution, denoted $St(\mathcal{I}) = (Sig',Sen',Mod',\sem{\_},\models)$, as follows: 
\begin{itemize}
\item $Sig'$ is the category, objects and morphisms of which are, respectively, 
quasi-representable signatures $\chi : \Sigma \to \Sigma'$,~\footnote{A signature morphism $\chi : \Sigma \to \Sigma'$ is
{\bf quasi-representable} if and only if each model homomorphism $h : Mod(\chi)(M') \to N$ has a unique
$\chi$-expansion $h' : M' \to N'$.} and pairs of base institution
signature morphisms  $(\varphi : \Sigma \to \Sigma_1,\varphi' : \Sigma' \to \Sigma'_1) : (\chi :
\Sigma \to \Sigma') \to  (\chi_1 : \Sigma_1 \to \Sigma'_1)$ such that:
$$\begin{CD}
\Sigma @>\chi>> \Sigma' \\
   @V\varphi VV @VV\varphi'V \\
\Sigma_1 @>>\chi_1> \Sigma'_1
\end{CD}$$ 
\noindent
is a weak amalgamation square~\footnote{We refer the reader to~\cite{Dia08} for a definition of weak amalgamation square.},
\item $Sen' :  Sig' \to Set$ is the functor that maps
every $\chi : \Sigma \to  \Sigma'$ to $Sen(\Sigma')$,
\item $Mod' : {Sig'}^{op} \to Cat$ is the functor that
maps $\chi : \Sigma \to \Sigma'$ to $Mod(\Sigma)$, and
\item $\sem{\_}$ is the $Sig'$-indexed 
%\Isa{la notation $|Sig'|$ n'est pas definie et je crois qu'on ne s'en sert pas ailleurs. $\rightarrow \Sigma'$ ?} \Marc{Tout \`a fait. Par contre, je ne comprends pas ta remarque $\rightarrow \Sigma'$ ?} \Isa{c'est Sig' et pas $\Sigma'$....} 
family of functors
$\sem{\_}_\chi : Mod'(\chi) \to Set$ that maps every
$\chi$-model $M$ to its set of states $\sem{M}_\chi =
\{ M' \in Mod(\Sigma') \mid Mod(\chi)(M') = M \}$. 
\end{itemize}
Given $\chi : \Sigma \rightarrow \Sigma'$ and a $\chi$-model $M$,
for each state $M' \in \sem{M}_\chi$, we define the 
satisfaction of $\varphi \in Sen'(\chi)$ by $M$ at
$M'$, denoted $M \models_\chi^{M'} \varphi$, by:
$$M \models_\chi^{M'} \varphi \mbox{ \ iff \ }
 M' \models_{\Sigma'} \varphi$$
Finally, a $\chi$-model $M$ satisfies $\varphi$, denoted
$M \models_\chi \varphi$ if and only if $M \models_\chi^{M'} \varphi$ for every
$M' \in \sem{M}_\chi$. \\ $St(\mathcal{I})$ is a stratified
institution where the concrete category is $Set$. Indeed, for each signature morphism $(\varphi,\varphi' :
(\chi :  \Sigma \to \Sigma') \to (\chi_1 : \Sigma_1
\to \Sigma'_1))$, the natural transformation
$\sem{\_}_{(\varphi,\varphi')}$ is defined by
$\sem{M}_{(\varphi,\varphi')}(M') = Mod(\varphi')(M')$
for each state $M' \in \sem{M}_{\chi'}$.
The definition of $\sem{\_}_{\chi}$ on model homomorphisms uses the
quasi-representable property of $\chi$.
The surjectivity of $\sem{\_}_{(\varphi,\varphi')}$ is assured
by the weak amalgamation property of the square defining
$(\varphi,\varphi')$.
\end{example}

\begin{example}
{\bf MPL} is the stratified institution where the concrete category is $Graph$, $\sem{(I,W,R)}_P = (I,R)$ for each
set $P$ of propositional variables and each $P$-model $(I,W,R)$, and for each signature morphism $\sigma : P \to P'$ and each $P'$-model $(I',W',R')$, $\sem{(I',W',R')}_\sigma$ is simply the identity morphism on $(I',R')$.  
\end{example}

\begin{example}
{\bf TMPL} is a stratified institution which follows the same definition as {\bf MPL} by replacing $\sem{(I,W,R)}_P = (I,R)$ by $\sem{(X,\tau,\nu)}_P = (X,\tau)$.  Hence, the concrete category is the category of topological spaces $Top$.
\end{example}

\begin{proposition}[\cite{AD07}]
\label{is an institution}
Any stratified institution is an institution.
\end{proposition}
(The proof of Proposition~\ref{is an institution} is substantially similar to that given in~\cite{AD07}.)
  
%\Marc{Comme la d\'efinition des institutions stratifi\'ees est un peu diff\'erente de celle donn\'ee dans~\cite{AD07}, est-ce que l'on donne la preuve de la proposition ci-dessus ? Je crains que \c ca ne change pas grand chose dans la preuve.}
\medskip
By this proposition, we will also denote by $\mathcal{I}$ the generic stratified institution\\ $(Sig,Sen,Mod,\sem{\_},\models)$.

\section{Internal logic and extension to fuzzy case}
\label{extensions}

Here, we propose to define abstractly the important logic concepts of Boolean connectives, quantifiers, and fuzzy reasoning. By ``abstractly'' we mean independently of any stratified institution. Boolean connectives and quantifiers have already been defined internally to any institution~\cite{Dia08}. But institutions only consider sentences (i.e. closed formulas), and the institution {\bf MPL} does not have semantic negation, disjunction and implication connectives, as abstractly defined in institutions. Here, as the satisfaction of formulas is defined from model states, the standard Boolean connectives can be defined in stratified institutions more ``finely'' than in institutions, and allow stratified institutions such as {\bf MPL} to have all standard Boolean connectives. 

%To our knowledge, fix-point operators have not received attention in ``universal logics'' such as institutions or Tarskian logics~\cite{Tar56}. The reason is that fix-point operators at a semantic level require to explicitly deal with states that parametrize the satisfaction of formulas.  Here, we propose to define fix-point operators independently of a given logic in the manner of intedition rnal logic.  
Fuzzy (or many-valued) reasoning has already received an institutional semantics~\cite{Dia13,Dia14}. The approach proposed here is substantially similar to that proposed in~\cite{Dia13} although defined in the framework of stratified institutions.     

\subsection{Internal logic and quantifiers}
\label{Internal boolean and quantifiers}

\medskip
Let $\mathcal{I}$ be a stratified institution. Let $\Sigma$ be a signature of $\mathcal{I}$. Let $M$ be a $\Sigma$-model. A $\Sigma$-sentence $\varphi'$ is in $M$ a 
\begin{itemize}
\item {\bf semantic negation} of $\varphi$ when $\sem{M}_\Sigma(\varphi') = \sem{M}_\Sigma \setminus \sem{M}_\Sigma(\varphi)$;
\item {\bf semantic conjunction} of $\varphi_1$ and $\varphi_2$ when $\sem{M}_\Sigma(\varphi') = \sem{M}_\Sigma(\varphi_1) \cap \sem{M}_\Sigma(\varphi_2)$;
\item {\bf semantic disjunction} of $\varphi_1$ and $\varphi_2$ when $\sem{M}_\Sigma(\varphi') = \sem{M}_\Sigma(\varphi_1) \cup \sem{M}_\Sigma(\varphi_2)$;
\item {\bf semantic implication} of $\varphi_1$ and $\varphi_2$ when $\sem{M}_\Sigma(\varphi') =  (\sem{M}_\Sigma \setminus \sem{M}_\Sigma(\varphi_1)) \cup \sem{M}_\Sigma(\varphi_2)$.
\end{itemize}
A stratified institution $\mathcal{I}$ has (semantic) negation when each $\Sigma$-formula has a negation in each $\Sigma$-model. It has (semantic) conjunction (respectively disjunction and implication) when any two $\Sigma$-formulas have a conjunction (respectively disjunction and implication) in each $\Sigma$-model. As usual, we note negation, conjunction, disjunction and implication by $\neg$, $\wedge$, $\vee$ and $\Rightarrow$, respectively. Unlike institutions that deal with sentences, stratified institutions such as {\bf MPL}, {\bf MMPL} and {\bf TMPL} have now semantic negation, disjunction and implication. 

\medskip
In the same way, it is equally easy to introduce abstract quantifiers in stratified institutions by following the same construction as in the definition of internal stratification given in Example~\ref{internal stratification}. Hence, let $\mathcal{I} = (Sig,Sen,Mod,\sem{\_},\models)$ be a stratified institution, let $\chi : \Sigma \to \Sigma'$ be a signature morphism in $Sig$ and let $M \in Mod(\Sigma)$ be a model. Then, $M \models^\eta_\Sigma (\forall \chi) \varphi$ if and only if for every $\Sigma'$-model $M'$ such that $Mod(\chi)(M') = M$ and every state $\eta' \in \sem{M'}_{\Sigma'}$ such that $\sem{M'}_\chi(\eta') = \eta$ we have that $M' \models^{\eta'}_{\Sigma'} \varphi$. Existential quantification is defined dually by replacing ``every model $M'$'' and ``every state $\eta'$" by ``some model $M'$'' and ``some state $\eta'$" in the definition of universal quantification. 
%\Marc{J'ai ajout\'e la condition sur les \'etats mais c'est \`a contr\^oler. Je pense que c'est ok pour satisfaire la condition $\exists \chi \varphi \Leftrightarrow \neg \forall \chi \neg \varphi$.} \Isa{Ca me semble OK aussi}

\subsection{Fuzzy case}

\subsubsection{Residuated lattice}

The algebraic structures underlying fuzzy logic are usually residuated lattices. Residuated lattices generalize Boolean algebras for classical logic by considering a set of truth values which may contain more than two values. 

\begin{definition}[Residuated lattice]
A {\bf residuated lattice} $(L,\bigwedge,\bigvee,\varotimes,\rightarrow,0,1)$ is:
\begin{itemize}
\item a bounded lattice $(L,\bigwedge,\bigvee,0,1)$ 
%where the ordering $\leq$ is defined using the operations $\bigwedge$ and $\bigvee$ as usual 
where $\bigwedge$ and $\bigvee$ are the supremum and infimum operators associated with a partial ordering $\leq$, and $0$ and $1$ are the least and the greatest elements, respectively;
\item $\varotimes$ and $\rightarrow$ are binary operators  such that:
\begin{itemize}
\item $(L,\varotimes,1)$ is a monoid, that is, $\varotimes$ is a commutative and associative operation with the identity $a \varotimes 1 = a$;
\item $\varotimes$ is isotone in both arguments;
\item the operation $\rightarrow$ is a residuation operation with respect to $\varotimes$, i.e.
$$a \varotimes b \leq c~\mbox{iff}~a \leq b \rightarrow c$$
\end{itemize}
\end{itemize}
\end{definition} 

Most famous examples of residuated lattices are Goguen algebra and Luckasiewicz algebra, defined respectively as follows:
\begin{itemize}
\item {\bf Goguen algebra.} $([0,1],\bigwedge,\bigvee,\varotimes,\rightarrow,0,1)$ where $\varotimes$ is the ordinary product of reals and 
$$a \rightarrow b = 
\left\{ 
\begin{array}{ll}
1 & \mbox{if $a \leq b$} \\
\frac{b}{a} & \mbox{otherwise} 
\end{array}
\right.$$
\item {\bf Luckasiewicz algebra.} $([0,1],\bigwedge,\bigvee,\varotimes,\rightarrow,0,1)$ where:
\begin{center}
$a \varotimes b = 0 \bigvee (a + b -1)$ \\
$a \rightarrow b = 1 \bigwedge (1 - a + b)$
\end{center}
\end{itemize}

%As we can observe in most of the examples of residuated lattices, the set of truth values is the closed interval $[0,1]$. In this case, the infimum and supremum can be abstractly defined from two binary operations $t : [0,1]^2 \to [0,1]$ and $s : [0,1]^2 \to [0,1]$ called t-norm and t-conorm, respectively. More formally, a t-norm $t : [0,1]^2 \to [0,1]$ is an operation which is commutative, associative, monotone, and satisfies the boundary condition $t(1,a) = t(a,1) = a$. Dually, a t-conorm $s : [0,1]^2 \to [0,1]$ is an operation which is commutative, associative, monotone, and satisfies the boundary condition $s(0,a) = s(a,0) = a$.  

\subsubsection{Institutional semantics}

Let $\mathcal{I} = (Sig,Sen,Mod,\sem{\_},\models)$ be a stratified institution. Let $\mathcal{L} = (L,\bigwedge,\bigvee,\varotimes,\rightarrow,0,1)$ be a residuated lattice. We can consider that for every signature $\Sigma$, the truth of $\Sigma$-formulas $\varphi \in Sen(\Sigma)$ is a value in $L$, i.e. for every $\Sigma$-model $M \in Mod(\Sigma)$, $\sem{M}_\Sigma(\varphi)$ is a fuzzy subset of $\sem{M}_\Sigma$ over $L$. Hence, whereas in $\mathcal{I}$, the satisfaction relation $M \models_\Sigma \varphi$ can be seen as a mapping from $\sem{M}_\Sigma$ to $\{0,1\}$, in a fuzzy extension of $\mathcal{I}$, $M \models_\Sigma \varphi$ is a mapping from $\sem{M}_\Sigma$ to $L$. For every $\eta \in \sem{M}_\Sigma$, we will rather use the notation $(M \models^\eta_\Sigma \varphi)$ than $M \models_\Sigma \varphi(\eta)$ to denote the value in $L$ yielded by the mapping $M \models_\Sigma \varphi$. Of course, to preserve the satisfaction condition, we have to impose the following equivalence: for each signature morphism $\sigma : \Sigma \to \Sigma'$, every $\Sigma'$-model $M'$, every $\Sigma$-formula $\varphi$ and every $\eta' \in \sem{M'}_\Sigma$,
$$(M' \models^{\eta'}_{\Sigma'} Sen(\sigma)(\varphi)) = (Mod(\sigma)(M') \models^{\sem{M'}_\sigma(\eta')}_\Sigma \varphi)$$

Standardly, Boolean connectives and quantifiers can be internally defined in any fuzzy extension of a stratified institution $\mathcal{I}$. To give a meaning to negation, we suppose that $\mathcal{L}$ is with complements ($\bar{.}$). Hence, a $\Sigma$-sentence $\psi$ is, in a $\Sigma$-model $M$, a
\begin{itemize}
\item {\bf fuzzy semantic negation} of $\varphi$ when for every $\eta \in \sem{M}_\Sigma$, $(M \models^\eta_\Sigma \psi) = \overline{(M \models^\eta_\Sigma \varphi)}$;
\item {\bf fuzzy semantic conjunction} of $\varphi_1$ and $\varphi_2$ when for every $\eta \in \sem{M}_\Sigma$, $(M \models^\eta_\Sigma \psi) = (M \models^\eta_\Sigma \varphi_1) \bigwedge (M \models^\eta_\Sigma \varphi_2)$;
\item {\bf fuzzy semantic disjunction} of $\varphi_1$ and $\varphi_2$ when for every $\eta \in \sem{M}_\Sigma$,  $(M \models^\eta_\Sigma \psi) = (M \models^\eta_\Sigma \varphi_1) \bigvee (M \models^\eta_\Sigma \varphi_2)$;
\item {\bf fuzzy semantic implication} of $\varphi_1$ and $\varphi_2$ when for every $\eta \in \sem{M}_\Sigma$,  $(M \models^\eta_\Sigma \psi) = (M \models^\eta_\Sigma \varphi_1) \rightarrow (M \models^\eta_\Sigma \varphi_2)$.
\end{itemize}
The following connective $\varotimes$ is often added, the fuzzy semantics of which is:
$$\forall \eta \in \sem{M}_\Sigma,  (M \models^\eta_\Sigma \varphi_1 \varotimes \varphi_2) = ((M \models^\eta_\Sigma \varphi_1) \varotimes (M \models^\eta_\Sigma \varphi_2)).$$

First-order quantifiers can also be easily represented in a fuzzy way. Let $\chi : \Sigma \to \Sigma'$ be a signature morphism in $Sig$ and let $M \in Mod(\Sigma)$ be a model. A $\Sigma$-sentence $\varphi'$ is a {\bf (fuzzy semantic) universal $\chi$-quantification} in $M$ when for every $\eta \in \sem{M}_\Sigma$, $(M \models^\eta_\Sigma \varphi') = \bigwedge \{(M' \models^{\eta'}_\Sigma \varphi) \mid Mod(\chi)(M') = M~\mbox{and}~\sem{M'}_\chi(\eta') = \eta\}$. Existential quantification is defined dually by replacing the infimum $\bigwedge$ by the supremum $\bigvee$. In Section~\ref{extension to the fuzzy case}, we will give a more general definition which allows us to extend to the fuzzy case a large family of dual logical operators  such as modalities.

%\Isa{est-ce qu'il serait interessant (et possible) de definir les quantifications sur $M$ et sur $\eta$ separement ?}

Fuzzy logics allow us to reason about formulas according to uncertainty. This leads to extend the satisfaction relation $\models_\Sigma$ to a binary relation between models in $Mod(\Sigma)$ and couples in $Sen(\Sigma) \times L$ as follows:
$$M \models_\Sigma (\varphi,l) \Longleftrightarrow l \leq \bigwedge \{(M \models^\eta_\Sigma \varphi) \mid \eta \in \sem{M}_\Sigma\}$$
where $\leq$ is the ordering defined on $L$.  

\medskip
We have then the following result that proves that fuzzy extensions of stratified institutions are institutions.

\begin{proposition}
For every signature morphism $\sigma : \Sigma \to \Sigma'$, every $\Sigma'$-model $M'$ and couple $(\varphi,l) \in Sen(\Sigma) \times L$, we have:
$$M' \models_{\Sigma'} (Sen(\sigma)(\varphi),l) \Longleftrightarrow Mod(\sigma)(M') \models_\Sigma (\varphi,l).$$
\end{proposition}

\begin{proof}
By definition, we have that:  
$$(M' \models^{\eta'}_{\Sigma'} Sen(\sigma)(\varphi)) = (Mod(\sigma)(M') \models^{\sem{M'}_\sigma(\eta')}_\Sigma \varphi).$$ 
As $\sem{M'}_\sigma$ is surjective, we also have that: 
$$\bigwedge \{(M' \models^{\eta'}_{\Sigma'} Sen(\sigma)(\varphi)) \mid \eta' \in \sem{M'}_{\Sigma'}\} = \bigwedge \{(Mod(\sigma(M') \models^\eta_\Sigma \varphi) \mid \eta \in \sem{M}_\Sigma\},$$ 
and we can conclude that:
$$l \leq \bigwedge \{(M' \models^{\eta'}_{\Sigma'} Sen(\sigma)(\varphi)) \mid \eta' \in \sem{M'}_{\Sigma'}\}  \Longleftrightarrow$$ 
$$l \leq \bigwedge \{(Mod(\sigma(M') \models^\eta_\Sigma \varphi) \mid \eta \in \sem{M}_\Sigma\}.$$ 
\end{proof}

\section{Duality from morphological dilations and erosions in stratified institutions}
\label{dual connectives}

In this section, we show that mathematical morphology~\cite{BHR07,Ser82} can be used for defining systematically and uniformly the different logical concepts such as quantifiers and modalities. Indeed, we can observe that most of unary modalities and quantifiers have always a dual, and they commute with conjunction and disjunction. This then enables us to define such logic concepts via algebraic dilations and erosions based on the notion of adjunction. 

In the rest of the paper, we consider a stratified institution $\mathcal{I}$ which has conjunction, disjunction and negation.

\subsection{Lattice of formulas}
\label{lattice of formulas}

Let us first define the lattice $(Sen(\Sigma)_{/_{\equiv_M}},\preceq_M)$ where $M \in Mod(\Sigma)$. 
%on which  to apply algebraic dilations and erosions. 
In the following, we consider only finite sets of formulas.

\medskip
Let $M \in Mod(\Sigma)$ be a model. Considering the inclusion on the power set $\mathcal{P}(\sem{M}_\Sigma)$, the poset $(\mathcal{P}(\sem{M}_\Sigma),\subseteq)$ is a complete lattice. Similarly, a lattice is defined on  the set $Sen(\Sigma)_{/_{\equiv_M}}$ where $Sen(\Sigma)_{/_{\equiv_M}}$ is the quotient space of $Sen(\Sigma)$ by the equivalence relation $\equiv_M$ defined by:
$$\varphi \equiv_M \psi \Longleftrightarrow \sem{M}_\Sigma(\varphi) = \sem{M}_\Sigma(\psi).$$
This lattice is $(Sen(\Sigma)_{/_{\equiv_M}},\preceq_M)$ where $\preceq_M$  is the partial ordering defined by: 
$$\varphi \preceq_M \psi \Longleftrightarrow \sem{M}_\Sigma(\varphi) \subseteq \sem{M}_\Sigma(\psi).$$

Any finite subset $\{\varphi_i\}$ of $Sen(\Sigma)$ has as supremum $\bigvee \{\varphi_i\}$ and infimum $\bigwedge \{\varphi_i\}$, corresponding to union and intersection in the complete lattice $(\mathcal{P}(\sem{M}_\Sigma),\subseteq)$, and then, following the definitions given in Section~\ref{Internal boolean and quantifiers}, $\bigvee \{\varphi_i\}$ and $\bigwedge \{\varphi_i\}$ are the semantic disjunction $\bigvee_i \varphi_i$ and semantic conjunction $\bigwedge_i \varphi_i$ of the formulas in $\{\varphi_i\}$, respectively. Hence, $(Sen(\Sigma)_{/_{\equiv_M}},\preceq_M)$ is a bounded lattice. Greatest and least elements are respectively $\top$ and $\bot$, corresponding to equivalence classes of tautologies and antilogies. 

\subsection{Morphological dilations and erosions of formulas based on structuring elements}
\label{morphological dilations and erosions of formulas based on structuring elements}

\subsubsection{Definitions}

The most abstract way to define dilation and erosion is as follows. Let $(L, \preceq)$ and $(L', \preceq')$ be two lattices. An algebraic dilation is an operator $\delta : L \to L'$  that commutes with the supremum, and an erosion is an operator $\varepsilon : L' \to L$ that commutes with the infimum. It follows that both operators are increasing, $\delta$ preserves the least element $\bot$, and $\varepsilon$ preserves the greatest element $\top$. Now, in the practice of mathematical morphology, morphological operators are often defined on sets (i.e. $L$ and $L'$ are the powersets or finite powersets of given sets $E$ and $E'$, and often $E=E'$ and $L=L'$) through a structuring element designed in advance. Let us recall here the basic definitions of dilation and erosion $D_B$ and $E_B$ over sets, where $B$ is a set called structuring element. Let $X$ and $B$ be two subsets of a set $E$, endowed with a translation operator. The dilation and erosion of $X$ by the structuring element $B$, denoted respectively by $D_B(X)$ and $E_B(X)$, are defined as follows:
\begin{center}
$D_B(X) = \{x \in E \mid \check{B}_x \cap X \neq \emptyset\}$ \\
$E_B(X) = \{x \in E \mid B_x \subseteq X\}$
\end{center}
where $B_x$ denotes the translation of $B$ at $x$, and $\check{B}$ the symmetrical of $B$ with respect to the origin of space.\footnote{Mathematical morphology has been mainly applied in image processing. The set $E$ is then a vector space such as $\mathbb{R}^n$ or $\mathbb{Z}^n$, elements of which represent image points. In this case, $B_x = \{x + b \mid b \in B\}$ where $+$ is the vectorial sum.} Similarly, the structuring element $B$ can also be seen as a binary relation on the set $E$ as follows: $(x,y) \in B \Longleftrightarrow y \in B_x$~\cite{BHR07}. This is the way we will consider structuring elements in this paper.

The most important properties of dilation and erosion based on a structuring element are the following ones~\cite{BHR07,IGMI_NajTal10,Ser82}:
\begin{itemize}
\item {\bf Monotonicity:} if $X \subseteq Y$, then $D_B(X) \subseteq D_B(Y)$ and $E_B(X) \subseteq E_B(Y)$; if $B \subseteq B'$, then $D_B(X) \subseteq D_{B'}(X)$ and $E_{B'}(X) \subseteq E_B(X)$. 
\item If for every $x \in E$, $x \in B_x$ (and this condition is actually necessary and sufficient), then 
\begin{itemize}
\item {\bf $D_B$ is extensive:} $X \subseteq D_B(X)$;
\item {\bf $E_B$ is anti-extensive:} $E_B(X) \subseteq X$.
\end{itemize}
\item {\bf Commutativity:} $D_B(X \cup Y) = D_B(X) \cup D_B(Y)$ and $E_B(X \cap Y) = E_B(X) \cap E_B(Y)$. 
\item {\bf Adjunction:} $X \subseteq E_B(Y) \Leftrightarrow D_B(X) \subseteq Y$.
\item {\bf Duality:} $E_B(X) = [D_{\check{B}}(X^C)]^C$ where $\_^C$ is the set-theoretical complementation. 
\end{itemize}
Hence, $D_B$ and $E_B$ are particular cases of general algebraic dilation and erosion on the lattice $(\mathcal{P}(E),\subseteq)$. 

\medskip
In stratified institutions, given a $\Sigma$-model $M$, $\sem{M}_\Sigma$ is an element of a concrete category $\mathcal{C}$ (i.e. $\mathcal{C}$ comes with a faithful functor $\mathcal{U}$ such that $\mathcal{U}(\sem{M}_\Sigma)$ is a set~\footnote{Let us recall that for simplicity in the notations we use $\sem{M}_\Sigma$ to denote both the object in the concrete category $\mathcal{C}$ and the underlying set associated by the faithful functor $\mathcal{U}$.}). Therefore, let us suppose that for each model $M \in Mod(\Sigma)$, $\sem{M}_\Sigma$ is equipped with a structuring element $B$ (i.e. $\forall \eta \in \sem{M}_\Sigma$, $B_\eta \subseteq \sem{M}_\Sigma$) which represents a relationship between states, i.e. $\eta' \in B_\eta$ iff $\eta'$ satisfies some relationship to $\eta$ (see the next section to have examples of structuring elements for given stratified institutions), and $\check{B}_\eta$ is defined by $\eta' \in \check{B}_\eta \Leftrightarrow \eta \in B_{\eta'}$. 
%\Isa{faut-il dire explicitement $B_\eta \subseteq \sem{M}_\Sigma$ ? on l'utilise dans la suite.} \Marc{Oui, je l'ai ajout\'e ci-dessus} 
Drawing inspiration from Bloch \& al. in~\cite{Bloch02,BL02}, dilation and erosion of a formula $\varphi \in Sen(\Sigma)$ then give rise to two formulas $D_B(\varphi)$ and $E_B(\varphi)$ satisfying the following equivalences:
\begin{center}
$\begin{array}{ll}
M \models^\eta_\Sigma D_B(\varphi) & \Longleftrightarrow \check{B}_\eta \cap \{\eta' \in \sem{M}_\Sigma \mid M \models^{\eta'}_\Sigma \varphi\} \neq \emptyset \\
                                                & \Longleftrightarrow \exists \eta' \in \check{B}_\eta, M \models^{\eta'}_\Sigma \varphi \\
                                                & \Longleftrightarrow \check{B}_\eta \cap \sem{M}_\Sigma(\varphi) \neq \emptyset
\end{array}$ \\
$\begin{array}{ll}
M \models^\eta_\Sigma E_B(\varphi) & \Longleftrightarrow B_\eta \subseteq \{\eta' \in \sem{M}_\Sigma \mid M \models^{\eta'}_\Sigma \varphi\} \\
                                                & \Longleftrightarrow \forall \eta' \in B_\eta, M \models^{\eta'}_\Sigma \varphi \\
                                                & \Longleftrightarrow B_\eta \subseteq \sem{M}_\Sigma(\varphi)
\end{array}$
\end{center}

\subsubsection{Extension to the fuzzy case}
\label{extension to the fuzzy case}

From our extension of stratified institutions to fuzzy reasoning, we can also define fuzzy dilation and erosion of formulas based on structuring elements. Several definitions of mathematical morphology on fuzzy sets with fuzzy structuring elements have been proposed in the literature, since the early work in~\cite{Bae95,BM93} (see e.g~\cite{IB:FSS-09,BM95,NACH-00} for reviews). Here, we follow the approach developed in~\cite{IB:FSS-09} using conjunctions and implications in residuated lattices. Hence, given a $\Sigma$-model $M$ with a structuring element $B$ such that for every $\eta \in \sem{M}_\Sigma$, $B_\eta$ is a fuzzy set, the dilation of a fuzzy formula by $B$ is defined for every $\eta \in \sem{M}_\Sigma$ as follows:
$$(M \models^\eta_\Sigma D_B(\varphi)) = \bigvee \{\check{B}_\eta(\eta') \varotimes (M \models^{\eta'}_\Sigma \varphi) \mid \eta' \in \sem{M}_\Sigma\}.$$
The erosion of a fuzzy formula by $B$ is defined for every $\eta \in \sem{M}_\Sigma$ as follows:
$$(M \models^\eta_\Sigma E_B(\varphi)) = \bigwedge \{B_\eta(\eta') \rightarrow (M \models^{\eta'}_\Sigma \varphi) \mid \eta' \in \sem{M}_\Sigma\}.$$

If we note $\mathcal{F}(\sem{M}_\Sigma)$ the set of all fuzzy sets on $\sem{M}_\Sigma$, the couple $(\mathcal{F}(\sem{M}_\Sigma), \leq)$ where $\leq$ denotes the fuzzy inclusion, is a complete lattice. Therefore, we can consider the lattice $(Sen(\Sigma)/_{\equiv_M},\preceq_M)$ where $\equiv_M$ and $\preceq_M$ are the fuzzy extensions of the two relations $\equiv_M$ and $\preceq_M$ defined in Section~\ref{lattice of formulas}. Here again, it is easy to show that the fuzzy versions of $D_B$ and $E_B$ commute with union and intersection of fuzzy sets of states, respectively, i.e. for every $\varphi_1,\varphi_2 \in Sen(\Sigma)$, we have:
\begin{itemize}
\item $D_B(\varphi_1 \vee \varphi_2) \equiv_M D_B(\varphi_1) \vee D_B(\varphi_2)$, 
\item $E_B(\varphi_1 \wedge \varphi_2) \equiv_M E_B(\varphi_1) \wedge E_B(\varphi_2)$, 
\end{itemize}
and then, $D_B$ and $E_B$, interpreted in a fuzzy sets setting, are algebraic dilation and erosion, respectively. As for the crisp case, it is quite straightforward to show that these fuzzy dilation and erosion are monotonous, extensive and anti-extensive when $\eta \in B_\eta$, and dual (resp. adjoint) if $\varotimes$ and $\rightarrow$ are dual (resp. adjoint). 
%\Isa{Faut-il laisser ces propri\'et\'es ici ou les mettre dans la section 4.3 ?}

%\Marc{Au lieu de se contenter de la derni\`ere phrase ci-dessus, ne pourrions-nous pas donner les m\^eme propri\'et\'es que pour le cas non flou (monotonicit\'e, extensivit\'e, anti-extensivit\'e, dualit\'e et adjonction) ? Elles sont valides aussi je pense, non ?}

\subsubsection{Examples}
\label{examples of instances}

We show in this section that the two dual logical operators $E_B$ and $D_B$ can be instantiated to define both first-order quantifiers $\forall,\exists$ and modalities $\Box,\Diamond$. Moreover, from Section~\ref{extension to the fuzzy case}, all these operators can naturally be extended to fuzzy cases.

\paragraph{First-order quantifiers}

Let $St({\bf Fol})$ be the stratified institution of the first-order logic
%~\footnote{We refer the reader to~\cite{AD07} for a complete definition of $St(\mathcal{I})$ given an institution $\mathcal{I}$.}. 
Let $\chi : (S,F,P) \hookrightarrow (S,F \coprod X,P)$ be a signature and let $x$ be a variable in $X$. For every $(S,F,P)$-model $M$, let us define the structuring element $B^x$ as follows:
$$\forall M' \in \sem{M}_\Sigma, B^x_{M'} = \{M'' \in \sem{M}_\Sigma \mid \forall y \neq x \in X, y^{M''} = y^{M'}\},$$ 
i.e. the set of models identical to $M'$ on all variables except possibly $x$.
This structuring element is symmetrical (i.e. $M'' \in B^x_{M'} \Leftrightarrow M' \in B^x_{M''}$) and contains the origin (i.e. $M' \in B^x_{M'}$).

We can then define the first-order quantifiers $\forall x$ and $\exists x$ as erosion and dilation from $B^x$ as follows:
\begin{center}
$\forall x. \varphi \equiv  E_{B^x}(\varphi)$, \\
$\exists x. \varphi \equiv  D_{B^x}(\varphi)$.
\end{center}

More generally,  in any internal stratification $St(\mathcal{I})$ of an institution $\mathcal{I}$, both quantifiers $\forall \chi$ and $\exists \chi$ for a signature $\chi : \Sigma \to \Sigma'$ can be defined similarly. Indeed, for every $\chi$-model $M$, let us define the structuring element $B^\chi$ as follows:
$$\forall M' \in \sem{M}_\chi, B^\chi_{M'} = \sem{M}_\chi$$
Again, the structuring element is symmetrical and contains the origin, we then have:
\begin{center}
$\forall \chi.\varphi \equiv E_{B^\chi}(\varphi)$, \\
$\exists \chi.\varphi \equiv D_{B^\chi}(\varphi)$.
\end{center}

\paragraph{Modalities for Kripke models}

Let $\mathcal{I}$ be a stratified institution whose concrete category is $Graph$. Hence for each $\Sigma$-model $M$, $\sem{M}_\Sigma$ is a directed graph $(\sem{M}_\Sigma,R_M)$. Obviously, this accessibility relation $R_M$ naturally leads to the structuring element $B$ defined as follows:
$$R_M(\eta,\eta') \Longleftrightarrow \eta' \in B_\eta.$$

The modalities $\Box$ and $\Diamond$ are then defined as follows:~\footnote{Here, we consider the set $\check{B}$ to define dilation because the accessibility relation is not necessarily symmetrical.}
\begin{center}
$\Box \varphi \equiv  E_B(\varphi)$, \\
$\Diamond \varphi \equiv  D_{\check{B}}(\varphi)$. 
\end{center}

\subsubsection{Properties}
\label{properties}

The following properties are the direct extensions of properties of dilation and erosion on sets to formulas. 

\begin{itemize}
\item {\bf Monotonicity:} if $\varphi \preceq_M \psi$, then $D_B(\varphi) \preceq_M D_B(\psi)$ and $E_B(\varphi) \preceq_M E_B(\psi)$.
\item {\bf Extensivity of dilation:} $\varphi \preceq_M D_B(\varphi)$
and {\bf anti-extensivity of erosion:} $E_B(\varphi) \preceq_M \varphi$
if and only if for every $\eta \in \sem{M}_\Sigma$, $\eta \in B_\eta$.
%\item {\bf Commutativity} $D_B(\varphi_1 \vee \varphi_2) \equiv_M D_B(\varphi_1) \vee D_B(\varphi_2)$ and $E_B(\varphi_1 \wedge \varphi_2) \equiv_M E_B(\varphi_1) \wedge E_B(\varphi_2)$ 
\item {\bf Adjunction:} $\varphi \preceq_M E_B(\psi) \Leftrightarrow D_B(\varphi) \preceq_M \psi$.
\item {\bf Commutativity with supremum or infinum:} $D_B(\varphi_1 \vee \varphi_2) \equiv_M D_B(\varphi_1) \vee D_B(\varphi_2)$ and $E_B(\varphi_1 \wedge \varphi_2) \equiv_M E_B(\varphi_1) \wedge E_B(\varphi_2)$.
\item {\bf Duality:} $E_B(\varphi) \equiv_M \neg D_{\check{B}}(\neg \varphi)$. 
\end{itemize}
 
It follows that $D_B$ and $E_B$ are respectively algebraic dilation and erosion over $(Sen(\Sigma)_{/_{\equiv_M}},\preceq_M)$, i.e. in $(Sen(\Sigma)_{/_{\equiv_M}},\preceq_M)$ $D_B$ and $E_B$ commute with supremum and infimum, respectively.  
Moreover, by a standard result of mathematical morphology~\cite{BHR07}, $E_B$ (respectively $D_B$) is the unique erosion (respectively the unique dilation) associated with $D_B$ (respectively $E_B$) by the adjunction property. From standard results of mathematical morphology and the adjunction property, we also have the following properties:
\begin{corollary}
\label{preservation for states}
\mbox{}

\begin{itemize}
\item $E_B(\top) \equiv_M \top$
\item $D_B(\bot) \equiv_M \bot$ 
\item $\varphi \preceq_M E_B(D_B(\varphi))$ 
\item $D_B(E_B(\varphi)) \preceq_M \varphi$ 
\item $E_B(D_B(E_B(\varphi))) \equiv_M E_B(\varphi)$
\item $D_B(E_B(D_B(\varphi))) \equiv_M D_B(\varphi)$
\item $E_B(\varphi) \equiv_M \bigvee \{\psi \mid D_B(\psi) \preceq_M \varphi\}$
\item $D_B(\varphi) \equiv_M \bigwedge \{\psi \mid \varphi \preceq_M E_B(\psi)\}$
\end{itemize}
\end{corollary}
It follows that $E_B D_B$ (closing) and $D_B E_B$ (opening) are morphological filters (i.e. increasing and idempotent operators). Moreover, closing and opening are dual (i.e. $D_B(E_B(\varphi)) \equiv_M \neg E_{\check{B}}(D_{\check{B}}(\neg \varphi))$.

\begin{theorem}
\label{fundamental theorem}
The following properties are satisfied by dilation and erosion of formulas. Note that now properties are expressed independently of a model $M$.

\begin{enumerate}
\item $E_B(\top) \equiv \top$ and $D_B(\bot) \equiv \bot$.
\item $\varphi \models E_B(\varphi)$. 
\item If for every model $M \in Mod(\Sigma)$ and every $\eta \in \sem{M}_\Sigma$, $\eta \in B_\eta$, then $\varphi \models D_B(\varphi)$ and $E_B(\varphi) \models \varphi$.
\item$D_B(\varphi \vee \psi) \equiv D_B(\varphi) \vee D_B(\psi)$ and  $E_B(\varphi \wedge \psi) \equiv E_B(\varphi) \wedge E_B(\psi)$. Moreover, we have: $D_B(\varphi \wedge \psi) \models D_B(\varphi) \wedge D_B(\psi)$ and $E_B(\varphi) \vee E_B(\psi) \models E_B(\varphi \vee \psi)$.
\item $E_B(\varphi) \equiv \neg D_{\check{B}}(\neg \varphi)$, or dually $D_B(\varphi) \equiv \neg E_{\check{B}}(\neg \varphi)$.
\item If the stratified institution has implication, then 
\begin{enumerate}
\item $E_B(\varphi \Rightarrow \psi) \models E_B(\varphi) \Rightarrow E_B(\psi)$, 
\item $(E_B(\varphi) \Rightarrow D_B(\varphi)) \equiv \top$ if for every $M \in Mod(\Sigma)$ and every $\eta \in \sem{M}_\Sigma$, $B_\eta  \cap \check{B}_\eta \neq \emptyset$,
\item $(D_B(E_B(\varphi)) \Rightarrow \varphi) \equiv (\varphi \Rightarrow E_B(D_B(\varphi))) \equiv \top$.
\end{enumerate}
\end{enumerate}
\end{theorem}

\begin{proof}
\begin{enumerate}
\item These first two properties are obvious to check.
\item Let $M \models_\Sigma \varphi$. Let $\eta \in \sem{M}_\Sigma$ and let $\eta' \in B_\eta$. By hypothesis, $M \models^{\eta'}_\Sigma \varphi$, and then $M \models^\eta_\Sigma E_B(\varphi)$.
\item Let $M \models_\Sigma \varphi$. Let $\eta \in \sem{M}_\Sigma$. As $\eta \in B_\eta$, we directly deduce that $M \models_\Sigma D_B(\varphi)$. \\ Let $M \models_\Sigma E_B(\varphi)$. Let $\eta \in \sem{M}_\Sigma$. As $\eta \in B_\eta$, by hypothesis we have that $M \models^\eta_\Sigma \varphi$ whence we can conclude. 
\item Let $M \in Mod(D_B(\varphi \vee \psi))$. This means that for every $\eta \in \sem{M}_\Sigma$, there exists $\eta' \in B_\eta$ such that $M \models^{\eta'}_\Sigma (\varphi \vee \psi)$, and then $M \models^{\eta'}_\Sigma \varphi$ or $M \models^{\eta'}_\Sigma \psi$. From this, we can directly conclude that $M \models^\eta_\Sigma D_B(\varphi)$ or $M \models^\eta_\Sigma D_B(\psi)$, i.e. $M \models^\eta_\Sigma D_B(\varphi) \vee D_B(\psi)$. \\ Let $M \in Mod(D_B(\varphi) \vee D_B(\psi))$. This means that for every $\eta \in \sem{M}_\Sigma$, there exists $\eta' \in B_\eta$ such that $M \models^{\eta'}_\Sigma \varphi$ or $M \models^{\eta'}_\Sigma \psi$, and then $M \models^{\eta'}_\Sigma (\varphi \vee \psi)$. From this, we can directly conclude that $M \models^\eta_\Sigma D_B(\varphi \vee \psi)$. \\ The proof to show that $Mod(E_B(\varphi \wedge \psi)) = Mod(E_B(\varphi) \wedge E_B(\psi))$ is (relatively) similar. \\ 
Let $M \models_\Sigma D_B(\varphi \wedge \psi)$. Let $\eta \in \sem{M}_\Sigma$. By hypothesis, there exists $\eta' \in B_\eta$ such that $M \models^{\eta'}_\Sigma (\varphi \wedge \psi)$, and then $M \models^{\eta'}_\Sigma \varphi$ and $M \models^{\eta'}_\Sigma \psi$. Therefore, we can write that $M \models^\eta_\Sigma D_B(\varphi)$ and $M \models^\eta_\Sigma D_B(\psi)$, whence we directly have that $M \models^\eta_\Sigma (D_B(\varphi) \wedge D_B(\psi))$. The converse inequality if however not true and examples where the inequality is strict can be exhibited, similar to the ones in the set theoretical setting.\\ 
Let $M \models_\Sigma E_B(\varphi) \vee E_B(\psi)$. Let $\eta \in \sem{M}_\Sigma$ and let $\eta' \in B_\eta$. By hypothesis, we necessarily have that  $M \models^{\eta'}_\Sigma \varphi$ or $M \models^{\eta'}_\Sigma \psi$. Otherwise, we would have neither $M \models^\eta_\Sigma E_B(\varphi)$ nor $M \models^\eta_\Sigma E_B(\psi)$ which would be a contradiction. Hence $M \models^{\eta'}_\Sigma \varphi \vee \psi$, and $M \models^\eta_\Sigma E_B(\varphi \vee \psi)$. Again counter-examples showing that the converse is not true are easy to exhibit.
\item 
$\begin{array}[t]{ll}
M \models_\Sigma E_B(\varphi) & \Leftrightarrow \forall \eta \in \sem{M}_\Sigma, M \models^\eta_\Sigma E_B(\varphi) \\
                                                    & \Leftrightarrow \forall \eta \in \sem{M}_\Sigma, \forall \eta' \in B_\eta, M \models^{\eta'}_\Sigma \varphi \\
                                                    & \Leftrightarrow \forall \eta \in \sem{M}_\Sigma, \forall \eta' \in B_\eta, M {\not \models^{\eta'}_\Sigma} \neg \varphi \\
                                                    & \Leftrightarrow \forall \eta \in \sem{M}_\Sigma, M {\not \models^\eta_\Sigma} D_{\check{B}}(\neg \varphi) \\
                                                    & \Leftrightarrow \forall \eta \in \sem{M}_\Sigma, M \models^\eta_\Sigma \neg D_{\check{B}}(\neg \varphi) \\
\end{array}$
\item 
\begin{enumerate}
\item Let $M \models_\Sigma E_B(\varphi \Rightarrow \psi)$. Let $\eta \in \sem{M}_\Sigma$ such that $M \models^\eta_\Sigma E_B(\varphi)$. Let $\eta' \in B_\eta$. By hypothesis, $M \models^{\eta'}_\Sigma \varphi$, and then, as $M \models_\Sigma E_B(\varphi \Rightarrow \psi)$, we also have that $M \models^{\eta'}_\Sigma (\varphi \Rightarrow \psi)$, and $M \models^{\eta'}_\Sigma \psi$. 
\item Let $\eta\in \sem{M}_\Sigma$ such that $M \models^\eta_\Sigma E_B(\varphi)$. Let $\eta' \in B_\eta \cap \check{B}_\eta$ (by hypothesis this intersection is not empty). Then we have that $M \models^{\eta'}_\Sigma \varphi$ since $\eta' \in B_\eta$, and then $M \models^\eta_\Sigma D_B(\varphi)$ since $\eta' \in \check{B}_\eta$.
\item These properties come from the extensivity of closing and from the extensivity of opening, which hold for $\preceq_M$ (see Corollary~\ref{preservation for states}).
\end{enumerate}
\end{enumerate}
\end{proof}

\subsection{Dual logical operators as algebraic dilation and erosion} 
\label{dual logical operators as algebraic dilation and erosion}

In this section, we provide an algebraic view of dual dilation and erosion, without referring to any structuring element over the set $\sem{M}_\Sigma$.
%It may be that we cannot define a structuring element over the set $\sem{M}_\Sigma$ characterizing the behavior of the dual logical operators. Defining algebraic erosion and dilation on lattices allows precisely us not to refer to any structuring element. In this case, 
%We define dual logical operators $E$ and $D$  as algebraic erosion and dilation on the bounded lattice $(Sen(\Sigma)_{/_{\equiv_M}},\preceq_M)$.

\subsubsection{Definition}

\begin{definition}[Algebraic erosion and dilation]
Let $E$ and $D$ be two dual logical operators for $\mathcal{I}$, i.e. $E$ and $D$ satisfy the equation:
$$\forall M \in Mod(\Sigma), \forall \varphi \in Sen(\Sigma), E(\varphi) \equiv_M \neg D(\neg \varphi)$$

We will say that $E$ and $D$ are {\bf algebraic erosion and dilation} if they satisfy the two following equations: $\forall M \in Mod(\Sigma), \forall \varphi_1,\varphi_2 \in Sen(\Sigma)$,

\begin{enumerate}
\item $D(\varphi_1 \vee \varphi_2) \equiv_M D(\varphi_1) \vee D(\varphi_2)$;
\item $E(\varphi_1 \wedge \varphi_2) \equiv_M E(\varphi_1) \wedge E(\varphi_2)$.
\end{enumerate}
\end{definition}

By standard results of mathematical morphology, we then have the following properties:
\begin{proposition}
\mbox{}
\begin{itemize}
\item {\bf Monotonicity of $D$:} if $\varphi \preceq_M \psi$, then $D(\varphi) \preceq_M D(\psi)$;
\item {\bf Preservation of $\bot$ by $D$:} $D(\bot) \equiv_M \bot$;
\item {\bf Monotonicity of $E$:} if $\varphi \preceq_M \psi$, then $E(\varphi) \preceq_M E(\psi)$;
\item {\bf Preservation of $\top$ by $E$:} $E(\top) \equiv_M \top$;
\end{itemize}
\end{proposition}

Unlike dilation and erosion defined through structuring elements, the dual logical operators $E$ and $D$ defined as algebraic erosion and dilation do not form necessarily an adjunction (see Section~\ref{modalities for topos-models} for an example) which is expressed, when it holds, as follows:
$$\forall \varphi,\psi \in Sen(\Sigma), D(\varphi) \preceq_M \psi \Longleftrightarrow \varphi \preceq_M E(\psi)$$

When adjunction holds between $E$ and $D$, by standard results in mathematical morphology, the following properties are satisfied:

\begin{itemize}
\item $\varphi \preceq_M E(D(\varphi))$ (extensivity of $ED$); 
\item $D(E(\varphi)) \preceq_M \varphi$ (anti-extensivity of $DE$); 
\item $E(D(E(\varphi))) \equiv_M E(\varphi)$;
\item $D(E(D(\varphi))) \equiv_M D(\varphi)$;
\item $E(D(E(D(\varphi)))) \equiv_M E(D(\varphi))$;
\item $D(E(D(E(\varphi)))) \equiv_M D(E(\varphi))$.
\end{itemize} 

Some properties are preserved independently of a model $M$.

%\Isa{Ne faudrait-il pas donner ici l'equivalent de la def 5.4 (avec la dualite ?) sans faire reference aux $M$ ? Sinon ce qui suit n'est pas bien defini je crois.}

%\Marc{Je ne comprends pas bien tes questions. J'ai l'impression que nous les obtenons directement \`a partir de la D\'ef. 5.4 et de la propri\'et\'e simple \`a d\'emontrer $(\forall M \in Mod(\Sigma),\varphi \equiv_M \psi) \Longrightarrow \varphi \equiv \psi$.}

\begin{theorem}
The following properties are satisfied by dilation and erosion of formulas: 
\begin{itemize}
\item {\bf Duality:} $D(\varphi) \equiv \neg E(\neg \varphi)$. %\Isa{enlever si c'est dans la def, ou alors si c'est equivalent \`a $\forall M, E(\varphi) \equiv_M \neg D(\neg \varphi)$}
\item {\bf Commutativity:} $D(\varphi_1 \vee \varphi_2) \equiv D(\varphi_1) \vee D(\varphi_2)$ and $E(\varphi_1 \wedge \varphi_2) \equiv E(\varphi_1) \wedge E(\varphi_2)$.
\item {\bf Monotonicity:} if $\varphi \models \psi$, then $D(\varphi) \models D(\psi)$ and $E(\varphi) \models E(\psi)$.
\item {\bf Preservation:} $D(\bot) \equiv \bot$ and $E(\top) \equiv \top$.
\end{itemize}
\end{theorem}

\begin{proof}
Duality, commutativity and preservation are direct consequences of the fact that $(\forall M \in Mod(\Sigma),\varphi \equiv_M \psi) \Longrightarrow \varphi \equiv \psi$. To prove monotonicity, let us suppose that $\varphi \models \psi$. Therefore, for every $M \in Mod(\varphi)$ we have that $M \models_\Sigma \varphi$ and $M \models_\Sigma \psi$, and then for every $\eta \in \sem{M}_\Sigma$ we have $M \models^\eta_\Sigma \varphi$ and $M \models^\eta_\Sigma \psi$ i.e. $\sem{M}_\Sigma(\varphi) =  \sem{M}_\Sigma(\psi) = \sem{M}_\Sigma$, whence we conclude $\varphi \equiv_M \psi$. As $D$ is monotonous for $\equiv_M$, we then have that $D(\varphi) \equiv_M D(\psi)$. Hence, for every $\eta \in \sem{M}_\Sigma$, we have that $M \models^\eta_\Sigma D(\psi)$, and then $M \in Mod(D(\psi))$.  The reasoning for $E$ is similar. 
\end{proof}

%Here also all the results established in this section give rise to a set of axioms and inference rules that are sound {\em de facto}. 

\subsubsection{Example: modalities for topos-models}
\label{modalities for topos-models}

When the modalities $\Box$ and $\Diamond$ are interpreted topologically, they cannot be expressed as erosion and dilation based on a structuring element. The reason is the heterogeneity of elements used to express $M \models^\eta_\Sigma \Box \varphi$ where we quantify existentially over open sets and universally over elements in open sets. We might be tempted to define the modality $\Box$ by an erosion $E_B$ followed by a dilation $D_B$ (i.e. a morphological opening) where $B$ would be the structuring element defined as: $\forall \eta \in \sem{M}_\Sigma, B_\eta = \bigcup \{O \in \tau \mid \eta \in O\}$ where $M = (X,\tau,\nu)$ is a topos-model. The problem is that in this case we would quantify universally on open sets and not existentially. However, we have seen that $\sem{M}_\Sigma(\Box \varphi)$ and $\sem{M}_\Sigma(\Diamond \varphi)$ define topological interior and closure of  $\sem{M}_\Sigma(\varphi)$. It is well known that interior and closure commute with intersection and union, respectively. Moreover, they are dual operators.
%one of the other. 
Hence, $\Box$ and $\Diamond$ are algebraic erosion and dilation, respectively. Finally, $\Box$ is anti-extensive (and dually $\Diamond$ is extensive) for $\preceq_M$. Indeed, let $\eta \in \sem{M}_\Sigma(\Box \varphi)$ be a sate. This means that there exists an open set $O \in \tau$ such that $\eta \in O$ and for every $\eta' \in O$, $\eta' \in \sem{M}_\Sigma(\varphi)$. Hence, we necessarily have that $\eta \in  \sem{M}_\Sigma(\varphi)$.  We can also easily show that $\varphi \equiv \Box \varphi$.\footnote{Let us note that the equivalence $\varphi \equiv E(\varphi)$ is satisfied by all logics for which the satisfaction of formulas of the form $E(\varphi)$ requires that, for all models, the relation between states is reflexive, such as {\bf FOL}, {\bf MPL} with reflexive model, {\bf TMPL} and {\bf MMPL}. On the other hand, we do not have for every $M \in Mod(\Sigma)$ that $\varphi \preceq_M E(\varphi)$.} On the contrary, adjunction does not hold in general except under the (necessary and sufficient) condition that the underlying topology of topos-models satisfies that the closed sets defining formulas are precisely the open sets.

\begin{proposition}
Let $M = (X,\tau,\nu)$ be a topos-model over a signature $\Sigma$. Then, we have: $\forall \varphi,\psi \in Sen(\Sigma), \Diamond \varphi \preceq_M \psi \Longleftrightarrow \varphi \preceq_M \Box \psi$ if and only if for every $\varphi \in Sen(\Sigma)$, $\sem{M}_\Sigma(\varphi)$ is a closed set of $X$ is equivalent to $\sem{M}_\Sigma(\varphi)$ is an open set of $X$.
\end{proposition}

\begin{proof}
$\Longrightarrow$: Let us assume that $\forall \varphi,\psi \in Sen(\Sigma), \Diamond \varphi \preceq_M \psi \Longleftrightarrow \varphi \preceq_M \Box \psi$. Let $\varphi$ be a $\Sigma$-formula such that $\sem{M}_\Sigma(\varphi)$ is a closed set. We then have that $\sem{M}_\Sigma(\Diamond \varphi) = \sem{M}_\Sigma(\varphi)$, and then $\Diamond \varphi \preceq_M \varphi$. By applying the equivalence $\Diamond \varphi \preceq_M \psi \Longleftrightarrow \varphi \preceq_M \Box \psi$ to $\psi = \varphi$, we obtain that $\varphi \preceq_M \Box \varphi$. As $\Box$ is anti-extensive, we can then conclude that $\sem{M}_\Sigma(\varphi) = \sem{M}_\Sigma(\Box \varphi)$, and then $\sem{M}_\Sigma(\varphi)$ is open. Dually, applying this to the complement set allows us to conclude that all open sets of $X$ are closed. 

\medskip
$\Longleftarrow$: Let us assume that the closed sets of $X$ defining a formula are precisely the open sets of $X$. Let $\varphi$ and $\psi$ be two formulas such that $\Diamond \varphi \preceq_M \psi$. 
By monotonicity of $\Box$, we have that $\Box \Diamond \varphi \preceq_M \Box \psi$. Now, by definition of $\Diamond$, $\sem{M}_\Sigma(\Diamond \varphi)$ is open, and then closed by hypothesis. Hence, we have that $\sem{M}_\Sigma(\Box \Diamond \varphi) = \sem{M}_\Sigma(\Diamond \varphi)$. But, as $\Diamond$ is extensive, we have that $\varphi \preceq_M \Diamond \varphi$, whence we can conclude that $\varphi \preceq_M \Box \psi$. \\ Conversely, if $\varphi \preceq_M \Box \psi$, then by monotonicity of $\Diamond$, we have that $\Diamond \varphi \preceq_M \Diamond \Box \psi$. But, $\sem{M}_\Sigma(\Box \psi)$ is open and then by hypothesis closed. Hence, we have that $\sem{M}_\Sigma(\Diamond \Box \psi) = \sem{M}_\Sigma(\Box \psi)$. By anti-extensivity of $\Box$, we can directly conclude that $\Diamond \varphi \preceq_M \psi$. 
\end{proof}

\subsection{A sound and complete entailment system}
\label{about completeness}

In this section, we define the syntactic approach to truth for stratified institutions equipped with dual operators. This consists in establishing consequence relations $\vdash$, called {\em proofs}, between set of formulas and formulas. The syntactic approach of truth is then complementary to the semantic one represented by the semantic consequence $\models$. When we have that $\vdash \subseteq \models$, the syntactic approach is said {\em sound} and when we have the opposite inclusion, it is said {\em complete}. 
To obtain the result of completeness, we need to consider that formulas are built inductively from ``basic" formulas by applying iteratively Boolean connectives and a $I$-indexed family of dual operators $E^i$ and $D^i$ (resp. $E^i_B$ and $D^i_{\check{B}}$ when erosion and dilation are defined based on a structuring element $B$) for $i \in I$. In Sections~\ref{Internal boolean and quantifiers},~\ref{morphological dilations and erosions of formulas based on structuring elements} and~\ref{dual logical operators as algebraic dilation and erosion}, we have already given an abstract definition of Boolean connectives and of dual operators $E^i$ and $D^i$. It remains then to give an abstract definition of basic formulas.

\begin{definition}[Basic formulas]
\label{basic formulas}
A set of formulas $B \subseteq Sen(\Sigma)$ is {\bf basic} if there exists a $\Sigma$-model $M_B \in Mod(\Sigma)$ and a state $\eta \in \sem{M_B}_\Sigma$ such that for every $M \in Mod(\Sigma)$ and every $\eta' \in \sem{M}_\Sigma$, $M \models^{\eta'}_\Sigma B$ if and only if there exists a morphism $\mu_{\eta'} : M_B \to M$ such that $\sem{\mu_{\eta'}}_\Sigma(\eta) = \eta'$. \\ $M_B$ and $\eta$ are called {\bf basic model} and {\bf basic state} for $B$, respectively. 
\end{definition}

The notion of basic formulas has been first defined in~\cite{Dia08,GP10} but in institutions, and then for sentences (i.e. closed formulas). Here, to take into account open formulas, the definition of basic formulas involves states. 

\begin{proposition}
\label{examples of basic sets}
Any set of atomic formulas in {\bf PL}, {\bf FOL}, {\bf MPL}, {\bf TMPL} and {\bf MMPL} is basic. 
\end{proposition} 

\begin{proof}
\begin{itemize}
\item[{\bf PL.}] Let $P$ be a propositional signature. Let $B \subseteq P$. Let $M_B$ be the model that associates $1$ to any $p \in B$ and $0$ to any $p \in P \setminus B$. The choice of $\eta \in \sem{M_B}_P$ is obvious because $\sem{M_B}_P = \mathbb{1}$ (cf. Example~\ref{propositional logic}). \\ Let $M \in Mod(P)$ such that $M \models_P B$. This means that for every $p \in B$, $M(p) = 1$ whence we can concude that $M_B \leq M$ where $\leq$ is the partial ordering on models in $Mod(P)$. Conversely, let us suppose a morphism $\mu : M_B \to M$ (obviously, by the definition of models in {\bf PL}, we have that $\sem{\mu}_P(\mathbb{1}) = \mathbb{1}$). By hypothesis, we have that $M_B \leq M$ whence we can directly conclude that for every $p \in B$, $M(p) = 1$.
%\Isa{I $\rightarrow \mathbb{1}$ comme page 7 (any singleton)? A voir aussi dans ce qui suit.}  
\item[{\bf FOL}.] Let $\Sigma = (S,F,P)$ be a signature. Let $B$ be a set of atomic formulas over a set of variables $X$. Let us denote $M_B$ the $\Sigma$-model defined by:
\begin{itemize}
\item $\forall s \in S, M_{B_s} = T_F(X)_s$;
\item $\forall f : s_1 \times \ldots \times s_n \to s \in F, f^{M_B} : (t_1,\ldots,t_n) \mapsto f(t_1,\ldots,t_n)$;
\item $\forall p :  s_1 \times \ldots \times s_n \in P, p^{M_B} = \{(t_1,\ldots,t_n) \mid p(t_1,\ldots,t_n) \in B\}$.
\end{itemize}
Let us set $\eta$ the variable interpretation defined as $x \mapsto x$. \\ Let $M \in Mod(\Sigma)$ be a model and $\nu : X \to M$ be an interpretation such that $M \models^\nu_\Sigma B$. Therefore, we can define 
$\mu_\nu : \left\{ 
\begin{array}{l}
x \mapsto \nu(x) \\
f(t_1,\ldots,t_n) \mapsto f^M(\mu_\nu(t_1),\ldots,\mu_\nu(t_n)) 
\end{array}
\right.$ 
which is a morphism.
Obviously, we have that $\sem{\mu_\nu}_\Sigma(\eta) = \nu$. \\ Conversely, let us suppose a morphism $\mu : M_B \to M$ such that $\sem{\mu}_\Sigma(\eta) = \nu$. Let $p(t_1,\ldots,t_n) \in B$. As $\sem{\mu}_\Sigma(\eta) = \nu$, for every $t \in T_F(X)$, we have that $\mu(t) = \nu(t)$, and then, as $\mu$ is a morphism, we can conclude that $(\nu(t_1),\ldots,\nu(t_n)) \in p^M$. 
\item[{\bf MPL}.] Let $P$ be a propositional signature. Let $B$ be a subset of $P$. Let $M_B$ be the model defined by:
\begin{itemize}
\item $I = \mathbb{1}$ (any singleton);
\item $W^{\mathbb{1}} = B$;
\item $R = \emptyset$. 
\end{itemize}
Obviously, $\eta = \mathbb{1}$. Let $M = (I',W',R')$ be a $P$-model and let $i' \in I'$ be a state such that $M \models^{i'}_P B$. Let us define the morphism $\mu_{i'} : \mathbb{1} \mapsto i'$. Obviously, we have that $\sem{\mu_{i'}}_P(\mathbb{1}) = i'$. \\ Conversely, let us suppose a morphism $\mu : M_B \to M$ such that $\sem{\mu}_P(\mathbb{1}) = i'$. As $W^{\mathbb{1}} \subseteq W'^{i'}$, we directly have that $M \models^{i'}_P B$. \\ It is standard in modal logic to restrict the class of models to satisfy supplementary axioms. For instance, to satisfy $\Box \varphi \Rightarrow \varphi$, models have to be reflexive (i.e. the accessibility relation is reflexive). In this case, the basic model $M_B$ is defined as previously except that $R = \{(\mathbb{1},\mathbb{1})\}$.
\item[{\bf TMPL}.] Let $P$ be a propositional signature. Let $B \subseteq P$. Let us denote $M_B$ the $P$-model defined by:

\begin{itemize}
\item $X = \{B\}$;
\item $\tau = \{\emptyset,\{B\}\}$ (the topology is both discrete and trivial);
\item $\nu : p \mapsto \left\{
\begin{array}{ll}
\{B\} & \mbox{if $p \in B$} \\
\emptyset & \mbox{otherwise}
\end{array} \right.$
\end{itemize}
Let us set $\eta = B$.
%\Isa{je ne comprends pas cela. B n'est pas un etat ?} \Marc{Si, et c'est l'unique \'etat de l'ensemble $X$.}. \\ 
Let $M = (X',\tau,\nu')$ be a $P$-model and $x \in X'$ such that $M \models^x_P B$. Then, let us define the mapping $\mu_x : B \mapsto x$. Let us show that $\mu_x$ is a morphism. First, let us show that it is continuous. Let $O \in \tau'$ be an open set. Two possibilities can occur:
\begin{enumerate}
\item $x \in O$. In this case, $\mu^{-1}_x(O) = \{B\}$;
\item $x \notin O$. In this case, $\mu^{-1}_x(O) = \emptyset$.
\end{enumerate}
In both cases, $\mu^{-1}_x(O)$ is an open set, and then $\mu_x$ is continuous. Let $p \in P$. Here, two cases have to be considered:
\begin{enumerate}
\item $p \in B$. As $M \models^x_P B$, we have that $x \in \nu'(p)$, and then $\mu_x(\nu(p)) \subseteq \nu'(p)$;
\item $p \notin B$. By definition of $M_B$, $\nu(p) = \emptyset$, and then $\mu_x(\nu(p)) = \emptyset$. 
\end{enumerate}
Conversely, let us suppose a morphism $\mu : M_B \to M$ such that $\sem{\mu}_P(B) = x$. Let $p \in B$. As $\mu$ is a morphism, we have that $\mu(B) = x \in \nu'(p)$, and then $M \models^x_P B$. 
\item[{\bf MMPL}.] The construction of the model $M_B$ for the logic {\bf MMPL} is similar to that for {\bf TMPL}, as from any metric space a topology can be induced.
\end{itemize}
\end{proof}

Then, let us set the framework for this section.

\medskip
{\bf Framework:} we consider  a stratified institution $\mathcal{I}$ the functor $Sen$ of which has a subfunctor $Sen^{base} : Sig \to Set$ (i.e. $Sen^{base}(\Sigma) \subseteq Sen(\Sigma)$) such that for every signature $\Sigma \in Sig$:

\begin{itemize}
\item $Sen^{base}(\Sigma)$ is basic, and
\item $Sen(\Sigma)$ is inductively defined from $Sen^{base}(\Sigma)$ by applying Boolean connectives in $\{\wedge,\vee,\Rightarrow,\neg\}$ and a $I$-indexed family of dual operators $E^i$ and $D^i$ (resp. $E^i_B$ and $D^i_{\check{B}}$ when erosion and dilation are defined over a structuring element $B$) such that for each $i \in I$, $E^i$ and $D^i$ are anti-extensive and extensive, respectively, and for all $\varphi \in Sen(\Sigma)$, $\varphi \models E^i(\varphi)$.~\footnote{In modal logic, the proof systems satisfying such a condition are said normal.}
\end{itemize}
For all the examples of stratified institutions developed in this paper, we define the functor $Sen^{base}$ as the mapping which associates to any signature $\Sigma \in Sig$ the set of atomic formulas. In {\bf PL}, the family of dual operators is indexed by the emptyset. In {\bf FOL}, the family of dual operators is indexed by a set of variables $X$. Hence, in {\bf FOL}, $E^x$ and $D^x$ are respectively $\forall x$ and $\exists x$. In {\bf MPL}, {\bf TMPL} and {\bf MMPL}, the family is indexed by any singleton as we only consider the couple of dual operators $\Box$ and $\Diamond$. \\ 
We have seen for all the examples where the dual operators $E^i$ and $D^i$ are erosion and dilation based on a structuring element $B$ that they are anti-extensive and extensive if for every model $M \in Mod(\Sigma)$ and for every state $\eta \in \sem{M}_\Sigma$, we have $\eta \in B_\eta$. Hence, {\bf PL} and {\bf FOL}, as well as {\bf MPL} when the category of models is restricted to reflexive models, meet all the requirements of our framework. This is the same for {\bf TMPL} (and hence for {\bf MMPL}) as $\Box$ and $\Diamond$ define topological interior and closure which are known to be anti-extensive and extensive (see Section~\ref{modalities for topos-models}). \\ Finally, from Property 2 in Theorem~\ref{fundamental theorem}, the property $\varphi \models E^i(\varphi)$ is always satisfied when dual operators $E^i$ and $D^i$ are defined using a structuring element $B$, as in {\bf FOL} and {\bf MPL}. For {\bf TMPL} (and then {\bf MMPL}), we have also seen in Section~\ref{modalities for topos-models} that this last property holds.

\begin{definition}[Tautology instance]
We call {\bf tautology instance} any formula $\varphi \in Sen(\Sigma)$ such that there exists a propositional tautology $\psi$ (i.e. $\psi$ is a tautology in the logic {\bf PL})  the propositional variables of which are among $\{p_1,\ldots,p_n\}$ and $n$ formulas $\varphi_i \in Sen(\Sigma)$ such that $\varphi$ is obtained by replacing in $\psi$ all the occurrences of $p_i$ by $\varphi_i$ for $i \in \{1,\ldots,n\}$. 
\end{definition}

What justifies such a definition is the following result:

\begin{proposition}
Let $\psi$ be a propositional tautology the propositional variables of which are among $\{p_1,\ldots,p_n\}$. Let $\varphi_1,\ldots,\varphi_n \in Sen(\Sigma)$ be $n$ formulas. Then, the formula $\varphi$ in $Sen(\Sigma)$ obtained by replacing in $\psi$ all the occurences of $p_i$ by $\varphi_i$ for $i \in \{1,\ldots,n\}$ is a tautology, i.e. for every $M \in Mod(\Sigma)$, $\sem{M}_\Sigma(\varphi) = \sem{M}_\Sigma$. 
\end{proposition}

\begin{proof}
Let $M \in Mod(\Sigma)$ be a model. Let $\eta \in \sem{M}_\Sigma$ be a state. Let us define the propositional model $\nu$ in {\bf PL} by:
$$\nu : p_i \mapsto 
\left\{
\begin{array}{ll}
1 & \mbox{if $M \models^\eta_\Sigma \varphi_i$} \\
0 & \mbox{otherwise}
\end{array} 
\right.$$
By hypothesis, we have that $\nu \models \psi$, and then we can conclude that $M \models^\eta_\Sigma \varphi$. 
\end{proof}

The proof of completeness that we present here follows Henkin's method \cite{Henkin1949}.
% \Marc{Je ne sais pas quelle r\'ef\'erence donn\'ee car c'est un type de preuve classique que l'on trouve dans tous les livres de logiques avanc\'ees.}. 
This method relies on the proof that every consistent set of formulas has a model. This relies on the deduction theorem which is known to fail for modal logics except under some conditions (see~\cite{HN12}). Here, we give a condition based on the notion of ``invariant formula" that we define just below and which ensures the deduction theorem. This condition differs from that given in~\cite{HN12} in the sense that it is not about a restriction of the application of the inference rule {\em Necessity} (see below). As we will see later in this section, our condition will prove to be similar for {\bf MPL} and {\bf TMPL} (and then {\bf MMPL}) to change the definition of $\Gamma \vdash_\Sigma \varphi$ into: $\Gamma \vdash_\Sigma \varphi$ iff there exists a finite susbset $\{\varphi_1,\ldots,\varphi_n\} \subseteq \Gamma$ such that $\vdash_\Sigma \varphi_1 \wedge \ldots \wedge \varphi_n \Rightarrow \varphi$ (so-called {\em local derivation}). 

\begin{definition}[Invariant formula]
Let $\varphi \in Sen(\Sigma)$. $\varphi$ is said {\bf invariant} if: $\forall i \in I, \forall M \in Mod(\Sigma), \varphi \preceq_M E^i(\varphi)$.
\end{definition}

When $E^i$ and $D^i$ are erosion and dilation based on a structuring element $B$, it is easy to see that every formula $\varphi \in Sen(\Sigma)$ such that, for every $M \in Mod(\Sigma)$, $\sem{M}_\Sigma(\varphi)$ is equal to either $\sem{M}_\Sigma$ or $\emptyset$ is an invariant formula. Hence, in {\bf FOL}, all closed formulas (i.e. without free (unbound) variables) are invariant, and in {\bf MPL}, tautologies and antilogies are invariant formulas.  It is easy to see that when an invariant formula is a tautology or an antilogy, then so is its negation. \\ In {\bf TMPL} (and then {\bf MMPL}), all tautologies and antilogies are also invariant formulas.\footnote{Note that the name ``invariant'' was chosen since it also holds that $E^i(\varphi) \preceq_M \varphi$.}

\begin{definition}[Formula instance] Let $\varphi,\varphi' \in Sen(\Sigma)$. The formula $\varphi'$ is an {\bf instance of} $\varphi$ for $i \in I$ ($I$ is the index set of the family of the dual operators $E^i$ and $D^i$) if for every $M \in Mod(\Sigma)$, $E^i(\varphi) \preceq_M \varphi'$.
\end{definition}

Formula instance generalizes in stratified institution the concept of substitutions which are standard in first-order logics. Indeed, in {\bf FOL}, given a formula $\varphi$, we have for every variable $x \in X$ that $\forall x.\varphi \Rightarrow \varphi(x/t)$ is a tautology where $t \in T_F(X)$ and $\varphi(x/t)$ is the formula obtained from $\varphi$ by substituting every free occurence of $x$ by the term $t$. Of course, by the hypothesis that each $E^i$ is anti-extensive, $\varphi$ is always an instance of itself for $i \in I$.

\medskip
We then consider the following Hilbert-system for the stratified institution $\mathcal{I}$. 

\begin{itemize}
\item {\bf Axioms:}
\begin{itemize}
\item {\em Tautologies:} all tautology instances;
\item {\em Duality:} $E^i(\varphi) \Leftrightarrow \neg D^i(\neg \varphi)$;
\item {\em Distribution:} $E^i(\varphi \Rightarrow \psi) \Rightarrow E^i(\varphi) \Rightarrow E^i(\psi)$ (this axiom is called the Kripke schema);
\item {\em Instantiation:} $E^i(\varphi) \Rightarrow \varphi'$ when $\varphi'$ is an instance of $\varphi$ for $i \in I$; 
\item {\em Invariability:} $\varphi \Rightarrow E^i(\varphi)$ when $\varphi$ is an invariant formula.
\end{itemize}
\item {\bf Inference rules:}
\begin{itemize}
\item {\em Modus Ponens:} $\frac{\varphi \Rightarrow \psi~~~\varphi}{\psi}$;
\item {\em Necessity:} $\frac{\varphi}{E^i(\varphi)}$.
\end{itemize}
\end{itemize}

In modal logic, the inference rules and axioms given above define the system $T$. The systems $S4$, $B$ and $S5$ can be obtained by adding respectively the axioms written in our framework as follows:
\begin{itemize}
\item $E^i(\varphi) \Rightarrow E^i(E^i(\varphi))$ ($S4$),
\item $\varphi \Rightarrow E^i(D^i(\varphi))$ ($B$),
\item $D^i(\varphi) \Rightarrow E^i(D^i(\varphi))$ ($S5$),
\end{itemize}
In contrast, by imposing the anti-extensivity property, the systems $K$ and $D$ of the modal logic are not taken into account here.

\begin{definition}[Derivation]
A formula $\varphi \in Sen(\Sigma)$ is {\bf derivable} from a set of assumptions $\Gamma \subseteq Sen(\Sigma)$, written $\Gamma \vdash_\Sigma \varphi$, if $\varphi \in \Gamma$, or is one of the axioms, or follows from derivable formulas through applications of the inference rules.
\end{definition}

Hence, the proof system for $\mathcal{I}$ can be defined by the four following inference rules:
$$
\begin{array}{ll}
\frac{\varphi \in \Gamma}{\Gamma \vdash_\Sigma \varphi} & \frac{\varphi: \mbox{Axiom}}{\Gamma \vdash_\Sigma \varphi} \\
 \\
\frac{\Gamma \vdash_\Sigma \varphi~~~\Delta \vdash_\Sigma \varphi \Rightarrow \psi}{\Gamma \cup \Delta \vdash_\Sigma \psi} & \frac{\Gamma \vdash_\Sigma \varphi}{\Gamma \vdash_\Sigma E^i(\varphi)} 
\end{array}
 $$ 

\medskip
These inference rules give rise to an entailment system~\cite{Mes89}, i.e. a $Sig$-indexed family of binary relations $\vdash_\Sigma \subseteq \mathcal{P}(Sen(\Sigma)) \times Sen(\Sigma)$. Standardly, the $Sig$-indexed family $\{\vdash_\Sigma\}_{\Sigma \in Sig}$ satisfies the following properties:

\medskip
\begin{tabular}{l}
{\bf Transitivity} if $\Gamma \vdash_\Sigma \Gamma'$ and $\Gamma' \vdash_\Sigma \Gamma''$, then $\Gamma \vdash_\Sigma \Gamma''$; \\
{\bf Monotonicity} if $\Gamma \vdash_\Sigma \varphi$ and $\Gamma \subseteq \Gamma'$, then $\Gamma' \vdash_\Sigma \varphi$; \\
{\bf Compacity} if $\Gamma \vdash_\Sigma \varphi$, then there exists a finite subset $\Gamma_0$ of $\Gamma$ such that $\Gamma_0 \vdash_\Sigma \varphi$; \\
{\bf Translation} $\Gamma \vdash_\Sigma \varphi$, then $\forall \sigma : \Sigma \to \Sigma', \sigma(\Gamma') \vdash_{\Sigma'} \sigma(\varphi)$.
\end{tabular}
This system is enough to infer other properties of $E^i$ and $D^i$ such as the commutativity of $E^i$ (resp. $D^i$) with infimum (resp. supremum). Moreover, by the assumptions and the properties of dilation and erosion (see Sections~\ref{properties} and~\ref{dual logical operators as algebraic dilation and erosion}), the proof system defined above is sound, i.e. if $\Gamma \vdash_\Sigma \varphi$, then $\Gamma \models_\Sigma \varphi$. 
Finally, thanks to the condition of ``invariability" for formulas, we get the deduction theorem.

\begin{proposition}[Deduction theorem]
\label{deduction theorem}
Let $\Gamma \subseteq Sen(\Sigma)$ be a set of assumptions. If $\varphi$ is an invariant formula, then we have $\Gamma \cup \{\varphi\} \vdash_\Sigma \psi$ if and only if $\Gamma \vdash_\Sigma \varphi \Rightarrow \psi$. 
\end{proposition}

%\Isa{Pour que ce soit plus complet, je crois qu'il faudrait definit $\Gamma$ dans chaque proposition.} \Marc{C'est fait.}

\begin{proof}
The necessary condition is obvious and can be easily obtained by Modus Ponens. The sufficient condition is proved by induction on the given proof. The more difficult case is that where the last inference rule is Necessity. We then have that $\Gamma \cup \{\varphi\} \vdash_\Sigma E^i(\psi)$.  This means that $\Gamma \cup \{\varphi\} \vdash_\Sigma \psi$ previously in the proof, and then by the induction hypothesis we have that $\Gamma \vdash_\Sigma \varphi \Rightarrow \psi$. By Necessity, Distribution and Modus Ponens, we have that $\Gamma \vdash_\Sigma E^i(\varphi) \Rightarrow E^i(\psi)$. By the invariant axiom and the fact that $\varphi$ is an invariant formula, $\Gamma \vdash_\Sigma \varphi \Rightarrow E^i(\varphi)$, and then by transitivity, we can conclude that $\Gamma \vdash_\Sigma \varphi \Rightarrow E^i(\psi)$. 
\end{proof}

The following corollary justifies proof by reduction ad absurbum. 

\begin{corollary}
For every $\Gamma \subseteq Sen(\Sigma)$ and $\varphi \in Sen(\Sigma)$ such that $\neg \varphi$ is an invariant formula, we have that $\Gamma \vdash_\Sigma \varphi$ if and only if $\Gamma \cup \{\neg \varphi\}$ is inconsistent (i.e. for every formula $\psi \in Sen(\Sigma)$, $\Gamma \cup \{\neg \varphi\} \vdash_\Sigma \psi$ and $\Gamma \cup \{\neg \varphi\} \vdash_\Sigma \neg \psi$).
\end{corollary} 

\begin{proof}
The ``$\Rightarrow$" part is obvious. Let us prove the ``$\Leftarrow$" part. Let us suppose that $\Gamma \cup \{\neg \varphi\}$ is inconsistent. This then means that we have both $\Gamma \cup \{\neg \varphi\} \vdash_\Sigma \varphi$ and $\Gamma \cup \{\neg \varphi\} \vdash_\Sigma \neg \varphi$. As $\neg \varphi$ is an invariant formula by Proposition~\ref{deduction theorem} we can write that $\Gamma \vdash_\Sigma \neg \varphi \Rightarrow \varphi$. The formula $(\neg \varphi \Rightarrow \varphi) \Rightarrow \varphi$ is a tautology axiom, and then by Modus Ponens we have that $\Gamma \vdash_\Sigma \varphi$. 
\end{proof}

\begin{definition}[Maximal Consistence]
A set of formulas $\Gamma \subseteq Sen(\Sigma)$ is {\bf maximally consistent} if it is consistent and there is no consistent set of formulas properly containing $\Gamma$ (i.e. for each formula $\varphi \in Sen(\Sigma)$, either $\varphi \in \Gamma$ or $\neg \varphi \in \Gamma$, but not both).
\end{definition}

\begin{proposition}
\label{lindenbaum lemma}
Let $\Gamma \subseteq Sen(\Sigma)$ be a consistent set of formulas. There exists a maximally consistent set of formulas $\overline{\Gamma} \subseteq Sen(\Sigma)$  that contains $\Gamma$.
\end{proposition}

\begin{proof}
Let $S = \{\Gamma' \subseteq Sen(\Sigma) \mid \Gamma'~\mbox{is consistent and}~\Gamma \subseteq \Gamma'\}$. The poset $(S,\subseteq)$ is inductive. Therefore, by Zorn's lemma, $S$ has a maximal element $\overline{\Gamma}$. By definition of $S$, $\overline{\Gamma}$ is consistent and contains $\Gamma$. Moreover, it is maximal. Otherwise, there exists a formula $\varphi \in Sen(\Sigma)$ such that $\varphi \notin \overline{\Gamma}$. As $\overline{\Gamma}$ is maximal, this means that $\overline{\Gamma} \cup \{\varphi\}$ is inconsistent, and then $\overline{\Gamma} \cup \{\neg \varphi\}$ is. As $\overline{\Gamma}$ is maximal, we can conclude that $\neg \varphi \in \overline{\Gamma}$.   
\end{proof}

Proposition~\ref{lindenbaum lemma} is a quite direct generalization to stratified institutions of Lindenbaum's Lemma. To obtain our result of completeness, we need to impose the following condition:

\medskip
{\bf Assumption.}  For every basic set of formulas $B \subseteq Sen^{base}(\Sigma)$, there exists a basic model $M_B \in Mod(\Sigma)$ and a basic state $\eta \in \sem{M_B}_\Sigma$ for $B$ such that for every $i \in I$ ($I$ is the index-set of the dual operators $E^i$ and $D^i$) and every $\varphi \in Sen(\Sigma)$, there exists a subset $Inst_i(\varphi)$ of instances of $\varphi$ for $i$ satisfying :
\begin{enumerate}
\item for every $\varphi' \in Inst_i(\varphi)$, $|\varphi'| \leq |\varphi|$ where $|\varphi|$ and $|\varphi'|$ are the numbers of Boolean connectives and dual operators in $\varphi$ and $\varphi'$, and
\item  $(\forall \varphi' \in Inst_i(\varphi), M_B \models^\eta_\Sigma \varphi') \Longrightarrow M_B \models^\eta_\Sigma E^i(\varphi)$.
\end{enumerate}

\begin{proposition}
\label{assumption true}
All the couples $(M_B,\eta)$ defined in the proof of Proposition~\ref{examples of basic sets} for {\bf PL}, {\bf FOL}, {\bf MPL}, {\bf TMPL} and {\bf MMPL} satisfy such an assumption.  
\end{proposition}

\begin{proof}
The proof for {\bf PL} is obvious because the set of dual operators is empty (except the conjunction and disjunction which are assumed in the definition of the logic). For {\bf MPL}, {\bf TMPL} and {\bf MMPL}, as $\Box$ is anti-extensive, for every $\varphi \in Sen(\Sigma)$, we can set $Inst_\mathbb{1}(\varphi) = \{\varphi\}$ (let us recall that the index set for dual operators is here represented by the singleton with the unique element $\mathbb{1}$). 
The first condition of the assumption is obviously satisfied. Finally, as $\Box$ is anti-extensive, the accessibility relation is reflexive, and then if $M_B \models^\mathbb{1}_\Sigma \varphi$ in {\bf MPL} (resp. $M_B \models^B_\Sigma \varphi$ in {\bf TMPL} and {\bf MMPL}), then we necessary have that  $M_B \models^\mathbb{1}_\Sigma \Box \varphi$ in {\bf MPL} (resp. $M_B \models^B_\Sigma \Box \varphi$ in {\bf TMPL} and {\bf MMPL}). 
\\ In {\bf FOL}, given a variable $x \in X$, let us set $Inst_x(\varphi) = \{\varphi(x/t) \mid t \in T_F(X)\}$. Obviously, the first condition of the assumption is satisfied. 
Finally, if we suppose that $M_B \models^{Id}_\Sigma \varphi(x/t)$ for every $t \in T_F(X)$, then we have for each $\sigma : X \to T_F(X)$ such that for every $y \neq x \in X$, $\sigma(y) = y$ and $\sigma(x) = t$ that $M_B \models^\sigma_\Sigma \varphi$, whence we can conclude that $M_B \models^{Id}_\Sigma \forall x. \varphi$. 
\end{proof}

\begin{proposition}
\label{existence of a model}
Let assume that the assumption is satisfied. Then, for every maximal consistent set of formulas $\Gamma \subseteq Sen(\Sigma)$, there exists a $\Sigma$-model $M$ and a state $\eta \in \sem{M}_\Sigma$ such that $\Gamma = \{\varphi \mid M \models^\eta_\Sigma \varphi\}$.  
\end{proposition}

\begin{proof}
Let us denote $B = \Gamma \cap Sen^{base}(\Sigma)$. By definition of basic set of formulas, there exists a basic model $M_B$ and a state $\eta$ for $B$ that satisfy the assumption. Then, let us show by induction on the size of $\varphi$ that:
$$\Gamma \vdash \varphi \Longleftrightarrow M_B \models^\eta_\Sigma \varphi$$

The cases of basic formulas and Boolean connectives are easily provable. Then, let $\varphi$ be of the form $E^i(\psi)$. \\
($\Rightarrow$) Let us suppose that $\Gamma \vdash E^i(\psi)$. By Modus Ponens with $\Gamma \vdash_\Sigma E^i(\psi) \Rightarrow \psi'$ (Instantiation) where $\psi' \in Inst_x(\psi)$, we then have that $\Gamma \vdash \psi'$. By the first condition of the assumption, we can apply the induction hypothesis on every $\psi' \in Inst_x(\psi)$, and the we have that $M_B \models^\eta_\Sigma \psi'$, whence by the seconde condition of the assumption, we can conclude that $M_B \models^\eta_\Sigma E^i(\psi)$. \\
($\Leftarrow$) Let us suppose that $M_B \models^\eta_\Sigma E^i(\psi)$. By anti-extensivity of $E^i$, we then have that $M_B \models^\eta_\Sigma \psi$. By the induction hypothesis, we have that $\Gamma \vdash \psi$, and then by Necessity, $\Gamma \vdash E^i(\psi)$. 
\end{proof}

\begin{theorem}[Completeness]
Let assume that the assumption is satisfied. Then, for every $\Gamma \subseteq Sen(\Sigma)$ and every $\varphi \in Sen(\Sigma)$ such that $\neg \varphi$ is an invariant formula, we have that:
$$\Gamma \models \varphi \Longrightarrow \Gamma \vdash \varphi$$
\end{theorem}

\begin{proof}
If $\Gamma {\not \vdash} \varphi$, then $\Gamma \cup \{\neg \varphi\}$ is consistent. By Proposition~\ref{lindenbaum lemma}, there exists a maximal consistent set of formulas $\Gamma'$ that extends $\Gamma$, and then by Proposition~\ref{existence of a model}, there exists a model $M$ and a state $\eta \in \sem{M}_\Sigma$ such that $M \models^\eta_\Sigma \neg \varphi$, i.e. $M {\not \models^\eta_\Sigma} \varphi$. 
\end{proof}

\begin{corollary}
The inference rules for {\bf PL} is complete for any formulas. They are complete in {\bf FOL} for every closed formulas, and in {\bf MPL} and {\bf TMPL} for tautologies (and then so is for {\bf MMPL})
\end{corollary}
We find the standard results of completeness, among other to {\bf MPL} and {\bf TMPL} (and then {\bf MMPL}) where it is known that the completeness result holds for the local derivation (which amounts to demonstrate tautologies). More precisely, for {\bf MPL}, we have shown the completeness for the proof system known under the name $T$ and its extensions $S4$, $B$ and $S5$. On the contrary, as the anti-extensivity and extensivity properties of $E^i$ and $D^i$ are imposed (and then the accessibility relations are necessarily reflexive), the abstract proof given here cannot be instantiated to show the completeness result for the systems $K$ and $D$. For these two systems, we cannot use the model $M_B$ defined for the logic {\bf MPL} in the proof of Proposition~\ref{examples of basic sets}  to prove their incompleteness. We have to consider the canonical model for which the set of states is the whole set of sets of maximally consistent formulas. The problem is that such a model has no equivalent for {\bf PL} and {\bf FOL}. An open problem would be to see if there exists a general proof based on Henkin's method which works both for logics with dual operators which are extensive and anti-extensive, and for logics with dual operators which are not.

\medskip
Similar proofs of completeness have already been obtained in the framework of institutions but only for first-order logics~\cite{Pet07,GP10}. In~\cite{Pet07}, the author follows Henkin's method to prove his first-order completeness result while in~\cite{GP10}, the authors use forcing methods to extend their first completeness result to infinitary first-order logics. \\ Here, we have extended these first results by unifying, in the framework of stratified institutions, a completeness proof which works both for {\bf FOL} and the modal logics such as $T$, $S4$, $ B$ and $S5$, {\bf TMPL} and {\bf MMPL}.

\section{Towards applications in qualitative spatial reasoning}
\label{qualitative spatial reasoning}

When dealing with qualitative spatial reasoning, spatial relationships are usually classified into topological,
%~\cite{BD07,Blo99,CF97,CBGG97,RCC92,BB07}, 
metric or directional relations~\cite{Handbook07,KUIP-00}. In this section, we briefly show how such relations can be expressed in our framework.
%~\cite{Lig11,MM15,RN07}. 

%\Isa{y aurait-il moyen de parler d'entit\'es spatiales sans revenir aux points ?}

\subsection{Topological relationships}
\label{sec:topological relationships}

Topological approaches to qualitative spatial reasoning usually describe relationships between spatial regions. Two models have emerged to formalize topological spatial relations between spatial entities: RCC-8~\cite{RCC92} and 9-intersection~\cite{Egenhofer91,EF91}.

\subsubsection{RCC-8} 
RRC-8 is a first-order theory based on a primitive connectedness relation ${\bf C}$. From this binary relation ${\bf C}$, many other binary relations can be defined, among which 8 were identified as being of particular importance, via the definition of a parthood predicate {\bf P} defined from {\bf C}:
\begin{enumerate}
\item ${\bf DC}(X,Y)$ which means that $X$ is disconnected from $Y$;
\item ${\bf EC}(X,Y)$ which means that $X$ is externally connected to $Y$;
\item ${\bf PO}(X,Y)$ which means that $X$ partially overlaps $Y$;
\item ${\bf TPP}(X,Y)$ (resp. ${\bf TPPi}(X,Y)$) which means that $X$ (resp. $Y$) is a tangential proper part of $Y$ (resp. $X$);
\item ${\bf NTPP}(X,Y)$ (resp. ${\bf NTPPi}(X,Y)$) which means that $X$ (resp. $Y$) is a non-tangential proper part of $Y$ (resp. $X$);
\item ${\bf EQ}(X,Y)$ which means that $X$ is identical to $Y$. 
\end{enumerate} 
Here, given a stratified institution $\mathcal{I} = (Sig,Sen,Mod,\sem{\_},\models)$ and a model $M \in Mod(\Sigma)$, the elements in $\sem{M}_\Sigma$ are spatial entities, and then formulas are combinations of such entities. The model RCC-8 is a first-order theory which allows one to quantify on spatial entities. Here, this would amount to quantify on states which is not allowed by the langage. Following~\cite{AB02}, we introduce the modality $U$ and its dual $A$~\footnote{In~\cite{AB02}, authors use $E$. We prefer $A$ in order to avoid confusion with the notation for erosion.} the semantics of which is as follows: 
\begin{itemize}
\item $M \models^\eta_\Sigma U \varphi$ iff $\forall \eta' \in \sem{M}_\Sigma, M \models^{\eta'}_\Sigma \varphi$
\item $M \models^\eta_\Sigma A \varphi$ iff $\exists \eta' \in \sem{M}_\Sigma, M \models^{\eta'}_\Sigma \varphi$
\end{itemize}
Using these primitives connectors, following~\cite{Bloch02}, it is easy to define, independently of any stratified institution, simple relations such as inclusion, exclusion and intersection by using standard Boolean connectives in $\{\wedge,\vee,\Rightarrow,\neg\}$ and the modalities $U$ and $A$. Hence, the binary relations ${\bf C}$, ${\bf DC}$, ${\bf PO}$ and ${\bf EQ}$ can be expressed in our framework as follows, where $\varphi$ and $\psi$ are formulas that denote, respectively, the regions $X$ and $Y$:
\begin{itemize}
\item ${\bf C}(X,Y)$: $A(\varphi \wedge \psi)$;
\item ${\bf DC}(X,Y)$: $U(\neg \varphi \vee \neg \psi)$;
\item ${\bf PO}(X,Y)$: $A(\varphi \wedge \psi)$, $A(\varphi \wedge \neg \psi)$, and $A(\neg \varphi \wedge \psi)$;
\item ${\bf EQ}(X,Y)$: $\varphi \Leftrightarrow \psi$. 
\end{itemize}
The other relations can benefit from the morphological operators. For this, we suppose that the stratified institution $\mathcal{I}$ is equipped with two dual logical operators $E$ and $D$ defined as an erosion and a dilation on the lattice $(Sen(\Sigma)_{/_{\equiv_M}},\preceq_M)$ for every signature $\Sigma$ and every $\Sigma$-model $M$ such that $E$ and $D$ are anti-extensive and extensive, respectively, for the binary relation $\preceq_M$.  To define adjacency (or external connection) ${\bf EC}(X,Y)$ between two regions $X$ and $Y$, we can then consider that these regions do not intersect but as soon as one of them is dilated, it has a non-empty intersection with the other one. This can be expressed as:
%by the three following formulas: $\varphi$ and $\psi$ are formulas that denote, respectively, the regions $X$ and $Y$
\begin{itemize}
\item ${\bf EC}(X,Y)$: $\neg(\varphi \wedge \psi)$ and $A(D(\varphi) \wedge \psi)$ and $A(\varphi \wedge D(\psi))$.
\end{itemize}
Now, the fact that a region $X$ is a {\em tangential proper part of} a region $Y$ (i.e. ${\bf TPP}(X,Y)$) can be expressed by the fact that $X$ is included in $Y$ but the dilation of $X$ is not, i.e.: 
\begin{itemize}
\item ${\bf TPP}(X,Y)$: $\varphi \Rightarrow \psi$ and $A(D(\varphi) \wedge \neg \psi)$.
\end{itemize}
Similarly, the fact that a region $X$ is a {\em non-tangential proper part of} a region $Y$ (i.e. ${\bf NTPP}(X,Y)$) can be expressed as:
\begin{itemize}
\item ${\bf NTPP}(X,Y)$: $\varphi \Rightarrow \psi$ and $\varphi \Rightarrow E(\psi)$ (or equivalently, $D(\varphi) \Rightarrow \psi$).
\end{itemize}

\subsubsection{9-intersection} 
The 9-intersection model transforms the topological relationships between two spatial entities $X$ and $Y$ into a point-set topology problem. That is, the topological relations between two objects $X$ and $Y$ are defined in terms of the intersection of boundary, interior and exterior of $X$ and $Y$. Hence, the  9-intersection  model captures  the topological relation between two spatial entities $X$ and $Y$ based on the intersections of the three topological parts of $X$ and those of $Y$. These $3 \times 3$ types of intersections are concisely represented by the 9-intersection matrix:
$$\left(
\begin{array}{lll}
\delta X \cap \delta Y & \delta X \cap Y^o & \delta X \cap Y^- \\
X^o \cap \delta Y & X^o \cap Y^o & X^o \cap Y^- \\
X^- \cap \delta Y & X^- \cap Y^o & X^- \cap Y^-
\end{array}
\right) $$
where $\_^o$, $\_^-$ and $\delta \_$ denote the interior, the exterior and the boundary, respectively. \\ For any stratified institution the model of which are topos-model, these $3 \times 3$ types of intersections can be easily defined. Indeed, if we suppose that the two regions $X$ and $Y$ are denoted by the two formulas $\varphi$ and $\psi$, then 
\begin{itemize}
\item their interior are $\Box \varphi$ and $\Box \psi$, 
\item the exterior are $\neg \Diamond \varphi$ and $\neg \Diamond \psi$, and
\item their boundary are $\varphi \wedge \neg \Box \varphi$ and  $\psi \wedge \neg \Box \psi$, and in our framework $\Box$ and $\Diamond$ are algebraic erosion and dilation, respectively.
\end{itemize}

\subsection{Distances and directional relative position}

Here, we assume a stratified institution $\mathcal{I}$ such that 
\begin{itemize}
\item either the category of states is the category of metric spaces $Met$ and in this case $\mathcal{I}$ is equipped with two logical operators $E$ and $D$ defined as erosion and dilation on the lattice $(Sen(\Sigma)_{/_{\equiv_M}},\preceq_M)$ for every signature $\Sigma$ and every $\Sigma$-model $M$ such that $E$ and $D$ are anti-extensive and extensive, respectively, for the binary relation $\preceq_M$;
\item or $\mathcal{I}$ is equipped with two logical operators $E$ and $D$ defined as an erosion and dilation based on an elementary symmetrical structuring element $B$. 
\end{itemize}
In this last case, we can define a distance $d$ that can take different forms depending on the considered spatial domain, as follows:
\begin{itemize}
\item $\forall \eta, d(\eta, \eta) = 0$;
\item $\forall \eta, \eta', \eta \neq \eta', d(\eta, \eta') = 1 \mbox{ iff } \eta' \in B_\eta$,
\item $\forall \eta, \eta', d(\eta, \eta') = \inf_{\pi(\eta, \eta')} l(\pi)$, where $\pi(\eta, \eta')$ is a path from $\eta$ to $\eta'$, i.e. a sequence $\eta_0 = \eta, \eta_1, ... \eta_n=\eta'$ such that $\forall i = 0, ... n-1, d(\eta_i, \eta_{i+1}) = 1$, and $l(\pi)$ is the length of the path (i.e. for $\pi = \eta_0, \eta_1, ... \eta_n), l(\pi) = n = \sum_{i=0}^{n-1} d(\eta_i, \eta_{i+1})$).
\end{itemize}
By construction, $d$ defines a metric.

In both cases, we can define a distance to a formula for every model $M \in Mod(\Sigma)$ as done in the Euclidean space for a distance from a point to a compact set:
\[
d(\eta, \varphi) = \inf_{M\models^{\eta'}_{\Sigma} \varphi} d(\eta, \eta').
\]
Given two formulas $\varphi$ and $\varphi'$, their minimum $d_{\min}$ and Hausdorff $d_H$ distances can be derived as:
\[
d_{\min}(\varphi, \varphi') = \inf_{M\models^{\eta}_{\Sigma} \varphi} d(\eta, \varphi'),
\]
\[
d_H(\varphi, \varphi') = \max \left ( \sup_{M\models^{\eta}_{\Sigma} \varphi'} d(\eta, \varphi), \sup_{M\models^{\eta'}_{\Sigma} \varphi'} d(\eta', \varphi) \right ) .
\]
As in the Euclidean case, these two distances can be conveniently expressed in terms of mathematical morphology. Details for the logic {\bf PL} are given in~\cite{Bloch02}. Similarly, we have here:
\[
d_{\min}(\varphi, \varphi') \leq n \mbox{ iff } A(D^n(\varphi) \wedge \varphi'),
\]
where $D^0$ is the identity mapping, $D^1 = D$ and $D^n = D D^{n-1}$ for $n > 1$, and:
\[
d_H(\varphi, \varphi') \leq n \mbox{ iff } \varphi' \Rightarrow D_B^n(\varphi) \mbox{ and } \varphi \Rightarrow D_B^n(\varphi').
\]
As an example of the potential use of such links between distances and dilation in spatial reasoning, let us consider the example in~\cite{Bloch02}. If we are looking at an object represented by $\psi$ in an area which is at a distance in an interval $[n_1 , n_2 ]$ of a region represented by $\varphi$, this corresponds to a minimum distance greater than $n_1$ and to a Hausdorff distance less than $n_2$. 
Then we have to check the following relation:
\[
\psi \Rightarrow \neg D^{n_1}(\varphi) \wedge D^{n_2}(\varphi).
\]
This expresses in a symbolic way an imprecise knowledge about distances represented
as an interval. If we consider a fuzzy interval, this extends directly by means of fuzzy
dilation.
These expressions show how we can convert distance information, which is usually defined in an analytical way, into algebraic expressions through mathematical
morphology, and then into logical expressions through the proposed abstract dual operators based on dilation and erosion.

\medskip
Directional relations can be defined in a similar way in the proposed framework, extending directly the {\bf PL} case detailed in~\cite{Bloch02}. Here, $D^d$ denotes the dilation corresponding to a directional information in the direction $d$. Then assessing whether $\varphi'$ represents a region of space which is in direction $d$ with respect to the region represented by $\varphi$ amounts to check the following relation:
\[
\varphi' \Rightarrow D^d(\varphi).
\]

\section{Conclusion}

In this paper, we have shown that the abstract framework of stratified institutions allows for unified definitions of connectives, quantifiers and morphological operators. Morphological dilation and erosion are defined in this framework both algebraically as operators that commute with the supremum and infimum of the underlying lattices, and using structuring elements. The duality property is emphasized, as a common property of pairs of operators or modalities in several logics. The proposed abstract definitions and properties are then instantiated in different logics, such as propositional logic, first order logic, modal logics, fuzzy logics. Finally, they are used in qualitative spatial reasoning framework to define abstract topological, metric and directional relations. This is consistent with the common use of mathematical morphology to deal with spatial information.

Many perspectives are naturally occurring. First, the completeness result of this paper requires that the dual operators $E^i$ and $D^i$ are anti-extensive and extensive, respectively, which excludes the modal logics $D$ and $K$. As mentioned in Section~\ref{about completeness}, it would be interesting to see whether there exists a general proof based on Henkin's method which works both for logics with dual operators which are extensive and anti-extensive, and for logics with dual operators which are not. Another interesting perspective would be to extend our general completeness result to the fuzzy setting. 
Finally, future work will aim at further exploring the spatial reasoning aspects. Moreover, theoretical results on complexity and tractability could be explored.
%\Marc{Cette derni\`ere phrase est peut-\^etre \`a d\'evelopper ?}

\bibliographystyle{plain}
\bibliography{biblio}

\end{document}